\documentclass[12pt]{article}
\usepackage{verbatim}
\usepackage{amsfonts}
\usepackage{graphics}
\usepackage{amsmath}
\usepackage{times}
\usepackage{appendix}
\usepackage{color}
\usepackage{enumerate}
\usepackage{fancyhdr,latexsym,amsmath,amsfonts,amssymb,amsbsy,amsthm,url}
\usepackage[margin=0.5in,footskip=0.25in]{geometry}
\usepackage{graphics,graphicx,epsfig}
\usepackage{breqn}
\usepackage{caption,subcaption,float,subfloat}
\newtheorem{theorem}{Theorem}[]

\usepackage[english]{babel}
\usepackage[dvipsnames]{xcolor}
\usepackage{tikz}

\numberwithin{equation}{section}

\fancypagestyle{usual}{
  \lhead[\thepage]{}
  \chead[{\it \shortauthors}]{\sc \shorttitle}
  \rhead[]{\thepage}
  \lfoot[]{}
  \cfoot[]{}
  \rfoot[]{}
  
}

\makeatletter

\renewcommand{\section}{
  \@startsection
  {section}
  {1}
  {0pt}
  {1.1\baselineskip}
  {0.2\baselineskip}
  {\sc \centering}
}

\renewcommand{\subsection}{
  \@startsection
  {subsection}
  {1}
  {0pt}
  {1.1\baselineskip}
  {0.2\baselineskip}
  {\sc \centering}
}

\renewcommand{\subsubsection}{
  \@startsection
  {subsubsection}
  {1}
  {0pt}
  {1.1\baselineskip}
  {0.2\baselineskip}
  {\sc \centering}
}

\makeatother

\begin{document}

\title{\large\sc A quantitative study on the role of TKI combined with Wnt/$\beta$-catenin signaling and IFN-$\alpha$ in the treatment of CML through deterministic and stochastic approaches}
\normalsize
\author{\sc{Sonjoy Pan} \thanks{Indian Institute of Technology Guwahati, Guwahati-781039, Assam, India, e-mail: sonjoy.pan@iitg.ac.in}
\and \sc{Soumyendu Raha} \thanks{Indian Institute of Science Bangalore, Bengaluru-560012, Karnataka, India, e-mail: raha@iisc.ac.in}
\and \sc{Siddhartha P. Chakrabarty} \thanks{Indian Institute of Technology Guwahati, Guwahati-781039, Assam, India, e-mail: pratim@iitg.ac.in,
Phone: +91-361-2582606, Fax: +91-361-2582649}}
\date{}
\maketitle
\begin{abstract}

We propose deterministic and stochastic models for studying the pharmacokinetics of chronic myeloid leukemia (CML), upon administration of IFN-$\alpha$ (the traditional treatment for CML), TKI (the current frontline medication for CML) and Wnt/$\beta$-catenin signaling (the state-of-the art therapeutic breakthrough for CML). To the best of our knowledge, no mathematical model incorporating all these three therapeutic protocols are available in literature. Further, this work introduces a stochastic approach in the study of CML dynamics. The key contributions of this work are: (1) Determination of the patient condition, contingent upon the patient specific model parameters, which leads to prediction of the appropriate patient specific therapeutic dosage. (2) Addressing the question of how the dual therapy of TKI and Wnt/$\beta$-catenin signaling or triple combination of all three, offers potentially improved therapeutic responses, particularly in terms of reduced side effects of TKI or IFN-$\alpha$. (3) Prediction of the likelihood of CML extinction/remission based on the level of CML stem cells at detection.


\end{abstract}

\section{Introduction}
\label{CML_Introduction}

The complex physiological process of hematopoiesis involving the differentiation and regeneration of hematopoietic stem cells (HSC) arising in the bone marrow is extremely regulated and critical for maintaining the blood supply in the body and ensuring a balance between the differentiated cell lines of red blood cells (erythrocytes), white blood cells (leukocytes) and platelets \cite{Ainseba10,Clapp16,Colijn05}. Leukemia which is characterized through elevated blood cells (generally white blood cells) count is triggered by mutated blood cells violating the highly regulated process of hematopoiesis. It can be classified as either acute or chronic (contingent on the progression speed of the disease) and is either myeloid or lymphocytic (contingent on the maturity type of the cells) \cite{Clapp16,Radulescu16}. A vast majority of cases of chronic myeloid leukemia (CML), a myeloproliferative condition is attributed to a reciprocal translocation of one chromosome 9 and one chromosome 22, resulting in one chromosome 9 being longer and one chromosome 22 being shorter, than their respective normal counterparts \cite{Clapp16,Radulescu16,Sawyers99}. The occurrence of resulting chromosomal abnormality, called Philadelphia chromosome (denoted Ph) leads to the production of the BCR-ABL fusion oncogene. The BCR-ABL oncogene encodes for tyrosine kinase (a protein) that triggers the white blood cells to proliferate at an abnormal rate as compared to normal white blood cells \cite{Clapp16,Radulescu16,Besse17}. The progression of CML can typically be observed in three phases, namely, the chronic phase, the accelerated phase and finally the blast crisis, which is akin to acute leukemia \cite{Koch17}.

Traditionally, the standard therapy for CML includes interferon alfa-2a (IFN-$\alpha$) whose curative effectiveness was highly limited \cite{Mendrazitsky17,Guilhot07}. The enhanced life expectancy of CML patients, upon administration of IFN-$\alpha$ is believed to be achieved through activation of the effector immune cells \cite{Burchert05}, specifically the stimulation of the effector T-cells \cite{Montoya02}. At the turn of this century, a new line of CML treatment by way of tyrose kinase inhibitor (TKI), such as imatinib, dasatinib and nilotinib brought about a great leap in terms of therapeutic promise \cite{Clapp16}. Because the TKIs were very target specific, they targeted the CML cells by preventing the activation of proteins \cite{Clapp16}. TKIs limit the proliferation of the leukemic cells and causes their apoptosis, with minimal invasive effect on the healthy cells \cite{Besse17}. Imatinib's success in terms of remission and patient survival was reported to be 90\% for a 10-year period in patients receiving continuous treatment \cite{Deininger05}. The remission is determined on the basis of the ratio of BCR-ABL to a constant transcript, namely, BCR, ABL or GUS \cite{Clapp16}. Despite achieving remission and significantly increased life expectancy through TKIs, a truly complete cure remains elusive, due to non-responsiveness of quiescent leukemic stem cells (LSCs) to TKIs \cite{Gallipoli11,Rea12}. Remissions along with concurrent existence of low-levels of leukemic cell levels believed to be due to allogenic immune response to tumor, is pivotal for long term cure \cite{Besse18}. However, a combination therapy of imatinib and IFN-$\alpha$ provides improved response upon administration of IFN-$\alpha$ after remission induced by imatinib has taken place \cite{Mendrazitsky17}. The toxicity and the consequent side effects resulting from multi-drug combination can be addressed by optimal treatment scheduling \cite{He16}. In addition to IFN-$\alpha$, allogenic bone marrow transplant \cite{Clapp16,Kolb95} acted as the predominant treatment for CML, predating the introduction of TKIs. While the mechanism of both IFN-$\alpha$ and allogenic bone marrow transplant relies on activation of immune response, the practical usability of the latter is extremely limited. The effectiveness of IFN-$\alpha$ involves mechanism such as apoptosis and activation of immune cells \cite{Talpaz16}. As already observed, the TKIs despite achieving great therapeutic success, is ineffective against quiescent LSCs, while the IFN-$\alpha$ is effective against them. The complementary effectiveness of both renders them as a good choice for combination therapy \cite{Clapp16}. The TKIs drive the CML cells into remission, leaving it up to the IFN-$\alpha$ to act against the residual cells. Further, IFN-$\alpha$ is capable of driving the quiescent LSCs to a point where they can be targeted by the TKIs, which otherwise would not have been possible \cite{Essers09,Sato09}. The results from clinical trials after cessation of therapeutic treatment of imatinib also gives credence to the role of immune response in CML \cite{Clapp16}.

The drug resistance of LSCs to TKIs leads to the consideration of alternative therapeutic protocols. The cells of the bone marrow microenvironment (BMM) regulate and protect LSCs from the effect of TKIs, resulting in the limited efficacy of TKIs against LSCs, which is responsible for disease persistence or relapse after discontinuation of the therapy \cite{Zhang13B,Agarwal17}. The Wnt which are a family of secreted glycoproteins plays an important role in the regulation of cell development, differentiation, proliferation and death \cite{Miller01}. The Wnt signaling can be modified by a Porcupine acyl transferase inhibitor. Such an inhibitor, WNT974, in combination with TKI regulates the Wnt/$\beta$-catenin signaling and potentially inhibits the proliferation of CML stem and progenitor cells \cite{Zhang13B,Agarwal17,Luis12}. The combination of Wnt/$\beta$-catenin signaling and TKI inhibits the BMM protection of LSCs from TKI targeted therapy and could enhance the therapeutic response in CML patients.

We discuss below some of the relevant mathematical models for CML dynamics as a precursor to the presentation of our proposed model. The discussion will mostly be centered around the dynamics of CML and the consequent immune response. The modeling of leukemia through a diffusion model was proposed in \cite{Afenya98} and analyzed to study the spatial effect of leukemic cells. Accordingly, the populations of normal and leukemic cells were considered. The model is suggestive that there is eventual breakdown in the coexistence of normal and leukemic cells with eventual replacement of normal cells by the leukemic cells. The periodic CML (PCML) wherein oscillations are observed in the hematopoietic system is modeled in \cite{Colijn05} with the key conclusion being that the levels of proliferating stem cells is critical in triggering the process of PCML. The dynamics of hematopoiesis was modeled by incorporating two different types of cells, namely, hematopoietic stem cells and differentiated cells for the normal as well as cancerous cells \cite{Ainseba10,Radulescu14,Candea16}. The global analysis suggested \cite{Ainseba10} that coexistence of normal and cancerous cells is not possible in the long run, with eventual convergence either to the safe equilibrium or the blast equilibrium. A model incorporating delay equations and three types of cell division along with clinical implication is dealt with in \cite{Radulescu14} wherein the evolution of different cell line populations is sought to be captured through the CML model studied. Cell divisions occurring during cell evolution in CML is considered in \cite{Candea16} along with the competition between normal and leukemic cells. A model involving the CML cancer cells and the naive T-cells as well as the CML specific effector T-cells was analyzed in \cite{Moore04}. Moore and Li \cite{Moore04} considered two types of immune cells, namely, the naive T-cells and the effector T-cells. The naive T-cell are activated only if they are CML-specific, whereas the effector T-cells are (apriori) CML-specific and can act immediately. In case a CML-specific naive T-cell is activated, it binds to a peptide major histocompatibility complex (MHC) pair and the presence of costimulators will eventually result in the proliferation of the T-cells. As week-long proliferation window is followed by differentiation of effector cells capable of producing immune response to the CML antigen \cite{Moore04}. A sampling based analysis led to the model prediction in which the Gompertzian growth rate as well as the natural death rate of CML cells (without altering the other model parameters) are the key parameters for the control of CML progression. The CML model due to Moore and Li \cite{Moore04} is revisited in \cite{Krishchenko16} and a rigorous analysis of the global dynamics for the same is carried out. The authors established that the dynamics of the system is unstable around the tumor-free equilibrium and obtained the global stability conditions for the internal tumor equilibrium. A general model for CML to include normal, leukemic and resistant types of both stems cells and progenitor cells is considered in \cite{Helal15} and the global analysis for the disease free, healthy free and endemic equilibrium is discussed. Further, optimal control problems for several imatinib therapy scenarios and the suboptimal response in CML are studied. A mathematical model incorporating the combination treatment of imatinib to target the CML cells and IFN-$\alpha$ driven immunotherapy is discussed in \cite{Berezansky12}. The model involved the dynamics of the CML cells and the effector T-cells with a time varying delay of about a week for the therapeutic effect of IFN-$\alpha$ to kick in. The scheduling of the combination treatment of imatinib and IFN-$\alpha$ through an optimal control problem is presented in \cite{Mendrazitsky17}. The model used involves the CML cells and the greater effector T-cells, akin to \cite{Berezansky12}. The key finding of the work was the duration of administration of imatinib and the cessation of treatment to be possible only vis-a-vis IFN-$\alpha$ therapeutic treatment. In Besse et al. \cite{Besse18}, the authors reduced a model proposed in \cite{Clapp15} (for which no theoretical analysis was provided) to one which accounts for cycling stem cells, mature leukemic cells and the immune cells. This simplification was driven by determination of the key characteristics of the CML dynamics reported in \cite{Clapp15}. A bifurcation analysis based on the model system is suggestive of TKI treatment to be consistent with a treatment free low disease state being existent and stable.

The organization of the remainder of the paper is as follows. In Section \ref{CML_Deterministic_Model}, we present the deterministic model and its analysis. The formulation and analysis of the stochastic model are carried out in Section \ref{CML_Stochastic_Model}. This is followed by the numerical results and discussions in Section \ref{CML_Numerical_Results}.

\section{Deterministic Model}
\label{CML_Deterministic_Model}
We propose the following model, involving the CML stem cells, the mature CML cells and the immune response in presence of combination therapy of Wnt/$\beta$-catenin signaling, TKIs and IFN-$\alpha$:
\begin{eqnarray}
\label{CML_Mathematical_Model_Equations}
\frac{dC_{s}(t)}{dt} &=& rC_{s}(t)\left(1-\frac{C_{s}(t)}{K}\right)-d_{1}C_{s}(t)-\beta_{1}C_{s}(t), \nonumber \\
\frac{dC_{m}(t)}{dt} &=& \mu C_{s}(t)-d_{2}C_{m}(t)-\gamma C_{m}(t)T(t)-\beta_{2}C_{m}(t), \\
\frac{dT(t)}{dt} &=& \alpha C_{m}(t)T(t)-d_3T(t)+\beta_{3}T(t). \nonumber
\end{eqnarray}
Here $C_{s}(t)$ and $C_{m}(t)$ denote the concentrations of CML stem cells and mature CML cells, respectively, with $T(t)$ representing the number of effector T-cells cytotoxic to CML (CTL), all at time $t$. The first term in the equation for $C_{s}$ describes the logistic growth for the CML stem cells, with $r$ being the logistic growth rate and the constant $K$ being the carrying capacity, which represents the maximum possible concentration level of the CML stem cells that can be sustained. The natural death rate leading to the clearance of the CML stem cells is denoted by $d_{1}$. The inclusion of response to the targeted therapeutic protocol of Wnt/$\beta$-catenin signaling by a Porcupine acyl transferase inhibitor leading to the reduction in the level of CML stem cells is assumed to be taking place as a result of dosage $\beta_{1}$ of Wnt/$\beta$-catenin signaling. In the second equation which describes the dynamics of the mature CML cells, the production of the mature CML cells is assumed to happen as a result of differentiation of the CML stem cells at a rate $\mu$ accompanied by a natural cell death rate happening at a rate of $d_{2}$. The reduction of the mature CML cells is assumed to happen as a result of the influence of the CTL immune cells and also in response to the targeted treatment of TKIs (imatinib, nilotinib, or dasatinib). While the former is assumed to happen at a rate $\gamma$, the reduction of level of mature CML cells is assumed to take place as a result of TKI dosage of $\beta_{2}$. As CTL is an antigen-specific immune cell, we ignore its effect on CML stem cells. It is assumed that CTL does not have any effect on the CML stem cells, until they differentiate into mature ones. Finally, for the third equation, the production of the effector T-cells (CTLs) is driven by the immune response triggered by the presence of the mature CML cells, at a rate $\alpha$. The natural clearance of the CTLs is assumed to happen at a rate $d_{3}$. Further, the stimulation of CTL is also as a result of the therapeutic response to IFN-$\alpha$, whose administered dosage is assumed to be $\beta_{3}$. Note that, adding the decay term, $-\mu C_{s}$ (due to differentiation of the CML stem cells into mature ones), in the first equation of (\ref{CML_Mathematical_Model_Equations}) results in the equation becoming
\[\frac{dC_{s}(t)}{dt} = r'C_{s}(t)\left(1-\frac{C_{s}(t)}{K'}\right)-d_{1}C_{s}(t)-\beta_{1}C_{s}(t),\]
where $\displaystyle{r'=r-\mu}$ and $K'=K\left(1-\frac{\mu}{r}\right)$.
However, the form of the first equation of (\ref{CML_Mathematical_Model_Equations}) was motivated by the models in \cite{Besse18,Clapp15}.
From the biological perspective, all the model parameters have to be positive, and accordingly, the initial condition for the system (\ref{CML_Mathematical_Model_Equations}) is $\left(C_s(0),C_m(0),T(0)\right)\in \mathbb{R}_{+}^{3}$. We now define $d_{1}^{\prime}=d_{1}+\beta_{1}$, $d_{2}^{\prime}=d_{2}+\beta_{2}$ and $d_{3}^{\prime}=d_{3}-\beta_{3}$ and assume that $d_{3}^{\prime} > 0$. We now show the non-negativity and boundedness of the solution to the model system (\ref{CML_Mathematical_Model_Equations}).
\begin{theorem}
The solution to the system (\ref{CML_Mathematical_Model_Equations}) with initial condition in $\mathbb{R}_{+}^{3}$ exists in $\mathbb{R}_{+}^{3}$ and are ultimately bounded.
\end{theorem}
\begin{proof}
From the theory of functional differential equations \cite{Hale13}, it follows that the system (\ref{CML_Mathematical_Model_Equations}) with the non-negative initial condition admits a unique solution $(C_s(t),C_m(t),T(t))$. Further, from the system (\ref{CML_Mathematical_Model_Equations}), we obtain,
\[\frac{dC_s}{dt} \bigg|_{C_s=0}=0,~\frac{dC_m}{dt} \bigg|_{C_m=0}=\mu C_s~\text{and }\frac{dT}{dt} \bigg|_{T=0}=0.\]
Therefore, it follows that the solution $\left(C_s,C_m,T\right)$ starting from $\mathbb{R}_{+}^{3}$ exists in $\mathbb{R}_{+}^{3}$.

Now, from the first equation of the model system (\ref{CML_Mathematical_Model_Equations}), we obtain $\displaystyle{\frac{dC_{s}(t)}{dt}\le rC_{s}\left(1-\frac{C_{s}}{K}\right)}$. Therefore, $\displaystyle{\limsup\limits_{t\rightarrow\infty}C_{s}\le K}$. We now introduce a new variable defined as $\displaystyle{X(t)=C_{m}+\frac{\gamma}{\alpha}T}$. From the last two equations of (\ref{CML_Mathematical_Model_Equations}), we obtain
\begin{eqnarray*}
\frac{dX(t)}{dt}&=&sC_{s}-d_2^{\prime}C_{m}-\frac{\gamma_2d_3^{\prime}}{\alpha}T\\
&\le&sK-d_{23}^{\prime}\left(C_{m}+\frac{\gamma}{\alpha}T\right),~\text{where }d_{23}^{\prime}=\min\{d_{2}^{\prime},d_{3}^{\prime}\}\\
&=&sK-d_{23}^{\prime}X.
\end{eqnarray*}
Therefore, $\displaystyle{\limsup\limits_{t\rightarrow\infty}X\le \frac{sK}{d_{23}^{\prime}}}$. Hence, $\displaystyle{\limsup\limits_{t\rightarrow\infty}C_{m}\le \frac{sK}{d_{23}^{\prime}}}$ and $\displaystyle{\limsup\limits_{t\rightarrow\infty}T\le \frac{\alpha sK}{\gamma d_{23}^{\prime}}}$. Thus the system (\ref{CML_Mathematical_Model_Equations}) with the non-negative initial condition has a unique, non-negative and ultimately bounded solution in the positively invariant set
\[\displaystyle{\mathcal{D}=\left\{(C_{s}(t),C_{m}(t),T(t))\in \mathbb{R}_{+}^{3}~:~
0\le C_s(t) \le K,~ 0\le C_m(t)\le  \frac{sK}{d_{23}^{\prime}},~ 0\le T(t) \le \ \frac{\alpha sK}{\gamma d_{23}^{\prime}}\right\}.}\]
\end{proof}
We now present the stability analysis for the steady states of the deterministic model \eqref{CML_Mathematical_Model_Equations}, which admits three steady states, namely,
\begin{enumerate}
\item $\displaystyle{\widetilde{E}=(\widetilde{C}_{s},\widetilde{C}_{m},\widetilde{T})=(0,0,0)}$.
\item $\displaystyle{\bar{E}=(\bar{C}_{s},\bar{C}_{m},\bar{T})=\left(\frac{K(r-d_{1}^{\prime})}{r},\frac{\mu K(r-d_{1}^{\prime})}{rd_{2}^{\prime}},0\right)}$.
\item $\displaystyle{E^{*}=(C_{s}^{*},C_{m}^{*},T^{*})=\left(\frac{K(r-d_{1}^{\prime})}{r},\frac{d_{3}^{\prime}}{\alpha},
\frac{\alpha\mu K(r-d_{1}^{\prime})-rd_{2}^{\prime}d_{3}^{\prime}}{r\gamma d_{3}^{\prime}}\right)}$.
\end{enumerate}
We now define a threshold parameter expressed as $\displaystyle{P^{*}=\frac{\alpha\mu K(r-d_{1}^{\prime})}{rd_{2}^{\prime}d_{3}^{\prime}}}$. The equilibrium point $\bar{E}$ exists if $r>d_{1}^{\prime}$ and $E^{*}$ exists if $r>d_{1}^{\prime}$ and $P^{*}>1$. The Jacobian for the model system is given by
\[J=
\begin{bmatrix}
r\left(1-\frac{2C_{s}}{K}\right)-d_{1}^{\prime}-rT & 0 & -rC_{s}\\
\mu & -d_{2}^{\prime}-\gamma T & -\gamma C_{m}\\
0 & \alpha T & \alpha C_{m}-d_{3}^{\prime}
\end{bmatrix}.\]
At $\widetilde{E}$, the eigenvalues of the Jacobian $J$ are $x=-d_{2}^{\prime}<0,~x=-d_{3}^{\prime}<0$ and $x=r-d_{1}^{\prime}$, which implies that $\widetilde{E}$ is stable when
$r<d_{1}^{\prime}$ and unstable when $r>d_{1}^{\prime}$. For $\bar{E}$, we have $r>d_{1}^{\prime}$ and the eigenvalues of the Jacobian $J$ are $x=-d_{2}^{\prime}<0,~x=-r+d_{1}^{\prime}<0$ and $x=\frac{\alpha\mu k(r-d_{1}^{\prime})-rd_{2}^{\prime}d_{3}^{\prime}}{rd_{2}^{\prime}}<0$ which holds only if $P^{*}<1$. The eigenvalues of the Jacobian $J$ at $E^{*}$ are obtained by solving the characteristic equation given by $x^3+a_2x^2+a_1x+a_0=0$, where,
\begin{eqnarray*}
a_0&=&\frac{r\alpha \gamma}{K}C_s^{*}C_m^{*}T^{*}+\alpha \mu \gamma C_s^{*}T^{*},\\
a_1&=&\frac{rd_2^{\prime}}{K}C_s^{*}+\frac{r\gamma}{K}C_s^{*}T^{*}+\alpha\gamma C_m^{*}T^{*},\\
a_2&=&d_2^{\prime}+\gamma T^{*}+\frac{r}{K}C_s^{*}.
\end{eqnarray*}
Since $C_s^{*}, C_m^{*}, T^{*}> 0$ if $P^{*}>1$, therefore $a_0, a_1, a_2 >0$ and $a_1a_2-a_0=\frac{rC_s^{*}}{K}\left(d_2^{\prime}+\gamma T^{*}\right)\left(d_2^{\prime}+\gamma T^{*}+\frac{rC_s^{*}}{K}\right)>0$ if $P^{*}>1$. Thus, by the Routh-Hurwitz criteria, all the three eigenvalues of the Jacobian $J$ at $E^{*}$ are either negative or have negative real parts if $P_{*}>1$.

In summary, we can state the following theorem:
\begin{theorem}
For the model system \eqref{CML_Deterministic_Model}, the following results hold:
\begin{enumerate}[(a)]
\item The disease free equilibrium $\widetilde{E}$ is stable when $r<d_{1}^{\prime}$ and unstable when $r>d_{1}^{\prime}$.
\item The disease but CTL free equilibrium $\bar{E}$ is locally asymptotically stable when $r>d_{1}^{\prime}$ with $P^{*}<1$ and unstable when $P^{*}>1$.
\item The disease and CTL coexistent equilibrium $E^{*}$ is locally asymptotically stable when $r>d_{1}^{\prime}$ with $P^{*}>1$.
\end{enumerate}
\end{theorem}
The assumption $r>d_{1}^{\prime}$ leads to the disease free equilibrium $\widetilde{E}$ being unstable. However, in case of $r<d_{1}^{\prime}$, the disease free equilibrium $\widetilde{E}$ would be stable which is suggestive of remission or even cure as a result of the combined rate of natural death and Wnt/$\beta$-catenin signaling exceeding the growth rate of the CML stem cells. Further, in case a patient is in state $E^{*}$ (which corresponds to a more severe case than $\bar{E}$), an appropriate choice of dosage of $\beta_{2}$ can drive the system from $P^{*}>1$ to $P^{*}<1$, resulting in stability at $\bar{E}$. We now define
\[P^{**}=\frac{\alpha\mu K(r-d_{1}^{\prime})-rd_{2}d_{3}^{\prime}}{rd_{3}^{\prime}}.\]
The condition for the most severe disease equilibrium $E^{*}$ is stable if $P^{*}>1$ \textit{i.e.,} $\beta_{2}<P^{**}$. If $\beta_{2}$ is increased such that $\beta_{2}>P^{**}$ which is equivalent to $P^{*}<1$, then the system converges to the less severe disease equilibrium $\bar{E}$. Finally, if $\beta_{1}$ is pushed beyond $r-d_{1}$, then the CML may be sent to remission or cure.

\section{Stochastic Model}
\label{CML_Stochastic_Model}

Although the deterministic model analyzed in Section \ref{CML_Deterministic_Model}, provides useful insight on the dynamics of CML progression, it does not offer any description of the probability of extinction of the disease. Accordingly, the random fluctuations in the physiological factors, the complex mechanism of the in-host immune system and the random process of cell growth/death motivates the introduction of stochasticity in the model. Therefore, in order to introduce randomness and uncertainly in the model description, we derive a stochastic differential equation (SDE) model, based on the state changes during interaction among the populations of CML stem cells, mature CML cells and CTLs themselves (considered in formulation of the deterministic model (\ref{CML_Mathematical_Model_Equations})), and we formulate the SDE model following the approach in \cite{Ahmad09,Dana11}.

Let $C_s$, $C_m$ and $T$ represent three random variables denoted in vector form as $X=(C_s,C_m,T)^{\top}$. The changes in these variables in a small time interval $\Delta t$ is given by $\Delta X=X(t+\Delta t)-X(t)$. Based on the mechanism of the deterministic model, the system (\ref{CML_Mathematical_Model_Equations}) can be written as
\[\dot{X}=S\Phi, \]
where
\[S=\begin{bmatrix}
r-d_1^{\prime} & 0 & 0 & -\frac{r}{K} & 0\\
\mu & -d_2^{\prime} & 0 & 0 & -\gamma \\
0 & 0 & -d_3^{\prime} & 0 & \alpha
\end{bmatrix}~\text{and}~
\Phi=\begin{bmatrix}
C_s \\ C_m \\ T \\ C_s^2 \\ C_mT
\end{bmatrix}
\]
represent the coefficient matrix and interaction strategy among the populations, respectively. Therefore the diffusion matrix is given by
\[B=S\cdot\text{diag}(\sqrt{\Phi})=
\begin{bmatrix}
(r-d_1^{\prime})\sqrt{C_s} & 0 & 0 & -\frac{r}{K}C_s & 0\\
\mu\sqrt{C_s} & -d_2^{\prime}\sqrt{C_m} & 0 & 0 & -\gamma\sqrt{C_mT} \\
0 & 0 & -d_3^{\prime}\sqrt{T} & 0 & \alpha\sqrt{C_mT}
\end{bmatrix}.\]
Hence
\[BB^{\top}=
\begin{bmatrix}
(r-d_1^{\prime})^2C_s+\frac{r^2}{K^2}C_s^2 & \mu(r-d_1^{\prime})C_s & 0\\
\mu(r-d_1^{\prime})C_s & \mu^2C_s+{d_2^{\prime}}^2C_m+\gamma^2C_mT & -\alpha\gamma C_mT \\
0 & -\alpha\gamma C_mT & {d_3^{\prime}}^2T+\alpha^2C_mT
\end{bmatrix}.\]
In order to reduce the number of Wiener processes, we find an alternative diffusion matrix $H$ having less number of columns such that $BB^{\top}=HH^{\top}$. By the Cholesky decomposition method, we compute
\[H=\begin{bmatrix}
h_{11} & 0 & 0\\
h_{21} & h_{22} & 0\\
0 & h_{32} & h_{33}
\end{bmatrix},\]
where
\begin{eqnarray*}
&&h_{11}=\sqrt{(r-d_1^{\prime})^2C_s+\frac{r^2}{K^2}C_s^2},~h_{21}=\frac{\mu(r-d_1^{\prime})C_s}{h_{11}},~h_{22}=\sqrt{\mu^2C_s+{d_2^{\prime}}^2C_m+\gamma^2C_mT-h_{21}^2},\\
&&h_{32}=\frac{-\alpha\gamma C_mT}{h_{22}},~h_{33}=\sqrt{{d_3^{\prime}}^2T+\alpha^2C_mT-h_{32}^2}.
\end{eqnarray*}
Hence the It\^{o} SDE model corresponding to the system (\ref{CML_Mathematical_Model_Equations}) is
\[dX(t)=f(X(t))dt+H(X(t))dW(t),\]
where $W(t)=(W_1(t),W_2(t),W_3(t))^{\top}$ is a vector of three independent Wiener processes. Thus the SDE model can be explicitly expressed as follows
\begin{eqnarray}
\label{CML_SDE_Model_Equations}
dC_{s}(t) &=& \left[rC_{s}(t)\left(1-\frac{C_{s}(t)}{K}\right)-d^{\prime}_{1}C_{s}(t)\right]dt+h_{11}dW_1(t), \nonumber \\
dC_{m}(t) &=& \left[\mu C_{s}(t)-d^{\prime}_{2}C_{m}(t)-\gamma C_{m}(t)T(t)\right]dt+h_{21}dW_1(t)+h_{22}dW_2(t), \\
dT(t) &=& \left[\alpha C_{m}(t)T(t)-d^{\prime}_3T(t)\right]dt+h_{32}dW_2(t)+h_{33}dW_3(t),\nonumber
\end{eqnarray}
where $d^{\prime}_1=d_1+\beta_1,~d^{\prime}_2=d_2+\beta_2$ and $d^{\prime}_3=d_3-\beta_3$.

We will show that the solution of the stochastic system (\ref{CML_SDE_Model_Equations}) is non-negative and bounded on $n$-th moment as well as stochastically bounded, in the sense of probability.
\begin{theorem}
\label{CML_Theorem_SDE_Sol_Exist}
For any initial condition $\displaystyle{X_0=(C_s(0),C_m(0),T(0)) \in \mathbb{R}^3_+}$, there exists a unique solution $\displaystyle{X(t)=}$ $\displaystyle{\left(C_s(t),C_m(t),T(t)\right)\in \mathbb{R}^3_+}$ to the stochastic system \eqref{CML_SDE_Model_Equations} for all $t\ge 0$ almost surely (a.s.) (\textit{i.e.}, with probability $1$), where $\displaystyle{\mathbb{R}^3_+=}$ $\displaystyle{\left\{(x_1,x_2,x_3) : x_i> 0,~i=1,2,3\right\}}$.
\end{theorem}
\begin{proof}
It is easily seen that the coefficients of the system \eqref{CML_SDE_Model_Equations} satisfy the local Lipschitz condition \cite{Arnold74,Mao07}. Therefore, for any initial condition $X_0 \in \mathbb{R}^3_+$, there exists a unique local solution $X(t)\in \mathbb{R}^3_+$ to the system \eqref{CML_SDE_Model_Equations} for $\displaystyle{t\in [0,t_e),~t_e>0}$. In order to prove the solution is global, \textit{i.e.}, $\displaystyle{X(t)\in\mathbb{R}^3_+}$ a.s. for all $t\ge 0$, we have to show that $t_e=\infty$ a.s. We choose a sufficiently large number $p_0\ge 0$ such that each of $\displaystyle{C_s(0),C_m(0)~\text{and}~T(0)}$ lies in $\displaystyle{\left[\frac{1}{p_0},p_0\right]}$. We now define
\[t_p=\inf\left\{t\in [0,t_e) : \text{at least one element of}~ \{C_s(t),C_m(t),T(t)\}\notin \left(\frac{1}{p},p\right)\right\}~\text{for}~ p\ge p_0\]
and $\displaystyle{\inf\phi=\infty}$, where $\phi$ is an empty set. We observe that $t_p$ is increasing as $p$ increases and denote $\displaystyle{t_{\infty}=\lim_{p\to\infty}t_p}$. This implies $\displaystyle{t_{\infty}\le t_e}$ a.s. Therefore, in order to show that $t_e=\infty$ a.s., we prove that $t_{\infty}=\infty$ a.s. If possible, let $t_{\infty}<\infty$ a.s. Then there exists two constants $t_T>0$ and $\epsilon\in (0,1)$ such that $P\{t_{\infty}\le t_T\}>\epsilon$. Hence there exists an integer $p_1\ge p_0$ such that
\begin{equation}
\label{SDE_Th_Sol_1}
P\{t_p\le t_T\}\ge\epsilon,~ \forall ~p\ge p_1.
\end{equation}
We define a $\mathbb{C}^2$-function $V:\mathbb{R}^3_+ \to \mathbb{R}_+$ by
\[V(X(t))=\Big(C_s(t)+1-\ln C_s(t)\Big)+\frac{{d}^{\prime}_1}{\mu}\Big(C_m(t)+1-\ln C_m(t)\Big)+\frac{\gamma{d}^{\prime}_1}{\alpha\mu}\Big(T(t)+1-\ln T(t)\Big).\]
The function $V$ is positive definite, since $(x+1-\ln x)> 0,~ \forall ~x>0$.
Using It\^{o} formula, we calculate the differential of $V$ along the solution trajectories of the stochastic system \eqref{CML_SDE_Model_Equations} as follows
\begin{eqnarray*}
dV(X(t))&=&LV(X(t))dt+\left(1-\frac{1}{C_s(t)}\right)h_{11}(t)dW_1(t)\\
&&+~\frac{{d}^{\prime}_1}{\mu}\left(1-\frac{1}{C_m(t)}\right)\Big(h_{21}(t)dW_1(t)+h_{22}(t)dW_2(t)\Big)\\
&&+~\frac{\gamma{d}^{\prime}_1}{\alpha\mu}\left(1-\frac{1}{T(t)}\right)\Big(h_{32}(t)dW_2(t)+h_{33}(t)dW_3(t)\Big),
\end{eqnarray*}
where
\begin{eqnarray*}
LV(X(t))&=&\left(1-\frac{1}{C_s(t)}\right)\left(rC_{s}(t)-\frac{rC^2_{s}(t)}{K}-d^{\prime}_{1}C_{s}(t)\right)\\
&&+~\frac{{d}^{\prime}_1}{\mu}\left(1-\frac{1}{C_m(t)}\right)\Big(\mu C_{s}(t)-d^{\prime}_{2}C_{m}(t)-\gamma C_{m}(t)T(t)\Big)\\
&&+~\frac{\gamma{d}^{\prime}_1}{\alpha\mu}\left(1-\frac{1}{T(t)}\right)\Big(\alpha C_{m}(t)T(t)-d^{\prime}_3T(t)\Big)\\
&&+~\frac{1}{2}\left[\frac{h^2_{11}(t)}{C^2_s(t)}+\frac{d^{\prime}(h^2_{21}(t)+h^2_{22}(t))}{\mu C^2_m(t)}+\frac{\gamma d^{\prime}_1(h^2_{32}(t)+h^2_{33}(t))}{\alpha\mu T^2(t)}\right]\\
&=&\left(r+\frac{r}{K}\right)C_s(t)-\frac{r}{K}C^2_s(t)-\frac{d^{\prime}_1}{\mu}(d^{\prime}_2+\gamma)C_m(t)-d^{\prime}_1\frac{C_s(t)}{C_m(t)}\\
&&+~\frac{\gamma d^{\prime}_1}{\mu}T(t)-\frac{\gamma d^{\prime}_1d^{\prime}_3}{\alpha\mu}T(t)+d^{\prime}_1-r+\frac{d^{\prime}_1d^{\prime}_2}{\mu}+\frac{\gamma d^{\prime}_1d^{\prime}_3}{\alpha\mu}\\
&\le&d^{\prime}_1-r+\frac{d^{\prime}_1d^{\prime}_2}{\mu}+\frac{\gamma d^{\prime}_1d^{\prime}_3}{\alpha\mu}+\left(r+\frac{r}{K}\right)C_s(t)+\frac{\gamma d^{\prime}_1}{\mu}T(t)
\end{eqnarray*}
Let \[c_1=d^{\prime}_1-r+\frac{d^{\prime}_1d^{\prime}_2}{\mu}+\frac{\gamma d^{\prime}_1d^{\prime}_3}{\alpha\mu}~~~\text{and}~~~c_2=2\left(r+\frac{r}{K}\right)+\frac{2\gamma d^{\prime}_1}{\mu}.\]
Then using the relation $x\le 2(x+1-\ln x),~\forall ~x> 0$, we obtain
\[\left(r+\frac{r}{K}\right)C_s(t)+\frac{\gamma d^{\prime}_1}{\mu}T(t)\le c_2V(X(t)).\]
Hence
\begin{eqnarray*}
dV(X(t))&\le& \Big(c_1+c_2V(X(t))\Big)dt+\left(1-\frac{1}{C_s(t)}\right)h_{11}(t)dW_1(t)\\
&&+~\frac{{d}^{\prime}_1}{\mu}\left(1-\frac{1}{C_m(t)}\right)\Big(h_{21}(t)dW_1(t)+h_{22}(t)dW_2(t)\Big)\\
&&+~\frac{\gamma{d}^{\prime}_1}{\alpha\mu}\left(1-\frac{1}{T(t)}\right)\Big(h_{32}(t)dW_2(t)+h_{33}(t)dW_3(t)\Big)\\
&\le& c_3\Big(1+V(X(t))\Big)dt+\left(1-\frac{1}{C_s(t)}\right)h_{11}(t)dW_1(t)\\
&&+~\frac{{d}^{\prime}_1}{\mu}\left(1-\frac{1}{C_m(t)}\right)\Big(h_{21}(t)dW_1(t)+h_{22}(t)dW_2(t)\Big)\\
&&+~\frac{\gamma{d}^{\prime}_1}{\alpha\mu}\left(1-\frac{1}{T(t)}\right)\Big(h_{32}(t)dW_2(t)+h_{33}(t)dW_3(t)\Big),
\end{eqnarray*}
where $c_3=\max\{c_1,c_2\}$. By taking $t_1\le t_T$, we obtain
\begin{eqnarray*}
\int_{0}^{t_p\land t_1}dV(X(t)) &\le& c_3\int_{0}^{t_p\land t_1}\Big(1+V(X(t))\Big)dt+\int_{0}^{t_p\land t_1}\left(1-\frac{1}{C_s(t)}\right)h_{11}(t)dW_1(t)\\
&&+~\frac{{d}^{\prime}_1}{\mu}\int_{0}^{t_p\land t_1}\left(1-\frac{1}{C_m(t)}\right)\Big(h_{21}(t)dW_1(t)+h_{22}(t)dW_2(t)\Big)\\
&&+~\frac{\gamma{d}^{\prime}_1}{\alpha\mu}\int_{0}^{t_p\land t_1}\left(1-\frac{1}{T(t)}\right)\Big(h_{32}(t)dW_2(t)+h_{33}(t)dW_3(t)\Big),
\end{eqnarray*}
where $t_p\land t_1=\min\{t_p, t_1\}$. This results in
\begin{eqnarray*}	
V(X(t_p\land t_1)) &\le& V(X_0)+c_3\int_{0}^{t_p\land t_1}\Big(1+V(X(t))\Big)dt+\int_{0}^{t_p\land t_1}\left(1-\frac{1}{C_s(t)}\right)h_{11}(t)dW_1(t)\\
&&+~\frac{{d}^{\prime}_1}{\mu}\int_{0}^{t_p\land t_1}\left(1-\frac{1}{C_m(t)}\right)\Big(h_{21}(t)dW_1(t)+h_{22}(t)dW_2(t)\Big)\\
&&+~\frac{\gamma{d}^{\prime}_1}{\alpha\mu}\int_{0}^{t_p\land t_1}\left(1-\frac{1}{T(t)}\right)\Big(h_{32}(t)dW_2(t)+h_{33}(t)dW_3(t)\Big).
\end{eqnarray*}
Taking expectation on both sides and using Fubini's theorem as well as properties of It\^o integral, we obtain
\begin{eqnarray*}	
EV(X(t_p\land t_1)) &\le& V(X_0)+c_3E\int_{0}^{t_p\land t_1}\Big(1+V(X(t))\Big)dt\\
&\le& V(X_0)+c_3t_T+c_3\int_{0}^{t_p\land t_1}EV(X(t))dt \\
\end{eqnarray*}
Using the Gronwall's inequality, we obtain
\begin{equation}
\label{SDE_Th_Sol_2}	
EV(X(t_p\land t_T)) \le c_4,
\end{equation}
where $c_4=\Big(V(X_0)+c_3t_T\Big)e^{c_3t_T}$.
Let $\Omega_p=\{t_p\le t_T\}$ for $p\ge p_1$. Therefore, using \eqref{SDE_Th_Sol_1}, it follows that $P(\Omega_p)\ge \epsilon$. Note that for every $\omega\in\Omega_p$, there is at least one of $C_s(t_p,\omega),~C_m(t_p,\omega)$ and $T(t_p,\omega)$ which equal(s) either $p$ or $\displaystyle{\frac{1}{p}}$ and hence \[V(X(t_p,\omega))\ge c_5,\] where \[c_5=\min\Big\{k(p+1-\ln p),~k\left(\frac{1}{p}+1+\ln p\right)\Big\},~~\text{with}~~k=\min\left\{1,\frac{d^{\prime}_1}{\mu},\frac{\gamma d^{\prime}_1}{\alpha\mu}\right\}.\]
Therefore \eqref{SDE_Th_Sol_1} and \eqref{SDE_Th_Sol_2} give
\[c_4\ge E[1_{\Omega_p}(\omega)V(X(t_p,\omega))]\ge \epsilon c_5,\]
where $1_{\Omega_p}$ is the indicator function of $\Omega_p$ and defined by
\[1_{\Omega_p}(\omega)=\begin{cases}
1 &\quad\text{if } \omega\in\Omega_p,\\
0 &\quad\text{if } \omega\notin\Omega_p.\\
\end{cases}\]
Taking $p\to \infty$ leads to the contradiction $\infty >c_4=\infty$. Hence we can conclude that $t_{\infty}=\infty$ a.s. This completes the proof of Theorem \ref{CML_Theorem_SDE_Sol_Exist}.
\end{proof}

\begin{theorem}
\label{CML_Theorem_SDE_Moment_Bounded}
Let $\displaystyle{X(t)=\left(C_s(t),C_m(t),T(t)\right)\in \mathbb{R}^3_+}$ be a solution of the stochastic system \eqref{CML_SDE_Model_Equations} with an initial condition $\displaystyle{X_0=(C_s(0),C_m(0),T(0)) \in \mathbb{R}^3_+}$. Then
\[\limsup_{t\to\infty} E[C^n_s(t)]\le M_1(n),~~\limsup_{t\to\infty} E[C^n_m(t)]\le M_2(n)~~\text{and}~~\limsup_{t\to\infty} E[T^n(t)]\le M_3(n),\]
where $M_i(n),~i=1,2,3$ are finite positive quantities and $n\ge 1$.
\end{theorem}
\begin{proof}
Let $U_1(C_s(t))=[C_s(t)]^n,~n\ge 1$. Using It\^o formula, we obtain from \eqref{CML_SDE_Model_Equations},
\[dU_1(C_s(t))=nC^n_s(t)\left[r\left(1-\frac{C_s(t)}{K}\right)-d^{\prime}_1+\frac{(n-1)h^2_{11}(t)}{2C^2_s(t)}\right]dt+nC^{n-1}_s(t)h_{11}(t)dW_1(t)\]
Integrating both sides from $0$ to $t$, we obtain
\[C^n_s(t)=C^n_s(0)+n\int_0^t C^n_s(u)\left[r\left(1-\frac{C_s(u)}{K}\right)-d^{\prime}_1+\frac{(n-1)h^2_{11}(u)}{2C^2_s(u)}\right]du+n\int_0^t C^{n-1}_s(u)h_{11}(u)dW_1(u)\]
Taking expectation on both sides and using Fubini's theorem as well as properties of It\^o integral, we obtain
\[E[C^n_s(t)]=C^n_s(0)+n(r-d^{\prime}_1)\int_0^t E[C^n_s(u)]du-\frac{nr}{K}\int_0^t E[C^{n+1}_s(u)]du+\frac{1}{2}n(n-1)\int_0^t E[h^2_{11}(u)C^{n-2}_s(u)]du\]
Differentiating both sides with respect to $t$, we obtain
\[\frac{dE[C^n_s(t)]}{dt}=n(r-d^{\prime}_1)E[C^n_s(t)]-\frac{nr}{K}E[C^{n+1}_s(t)]+\frac{1}{2}n(n-1)\left\{(r-d^{\prime}_1)^2E[C^{n-1}_s(t)]+\frac{r^2}{K^2}E[C^{n}_s(t)]\right\}\]
By H\"older's inequality, it follows that \[E\left[x^{n+1}\right] \ge \Big\{E\left[x^n\right]\Big\}^{\frac{n+1}{n}}.\] Using the above inequality, one can obtain
\[\frac{dE[C^n_s(t)]}{dt} \le n\left[r-d^{\prime}_1+\frac{(n-1)r^2}{2K^2}\right] E[C^n_s(t)]-\frac{nr}{K}\Big\{ E[C^{n}_s(t)]\Big\}^{\frac{n+1}{n}}+\frac{1}{2}n(n-1)(r-d^{\prime}_1)^2\Big\{ E[C^{n}_s(t)]\Big\}^{\frac{n-1}{n}}\]
Letting $E[C^{n}_s(t)]\ge 1$ leads to
\[\frac{dE[C^n_s(t)]}{dt} \le n\left\{r-d^{\prime}_1+\frac{(n-1)}{2}\left[\frac{r^2}{K^2}+(r-d^{\prime})^2\right]\right\} E[C^n_s(t)]-\frac{nr}{K}\Big\{ E[C^{n}_s(t)]\Big\}^{\frac{n+1}{n}}\]
Using the technique of solving Bernoulli differential equations and theory of differential inequalities, we obtain
\begin{eqnarray}
\label{CML_C_s_Bound}
\limsup_{t\to\infty} E[C^n_s(t)]\le \left(\frac{A_1 K}{r}\right)^n:=M_1(n),
\end{eqnarray}
where
\begin{equation}
A_1 = r-d^{\prime}_1+\frac{(n-1)}{2}\left[\frac{r^2}{K^2}+(r-d^{\prime})^2\right].
\end{equation}
Further, let $U_2(C_m(t))=[C_m(t)]^n,~n\ge 1$. Using It\^o formula, we obtain from \eqref{CML_SDE_Model_Equations},
\begin{eqnarray*}
dU_2(C_m(t))&=&nC^{n-1}_m(t)\Big[\mu C_s(t)-d^{\prime}_2C_m(t)-\gamma C_m(t)T(t)\Big]dt+\frac{1}{2}n(n-1)C^{n-2}_m(t)\Big[h^2_{21}(t)+h^2_{22}(t)\Big]dt\\
&&+~nC^{n-1}_m(t)\Big[h_{21}(t)dW_1(t)+h_{22}(t)dW_2(t)\Big]
\end{eqnarray*}
Taking expectation after integration from $0$ to $t$, we obtain
\begin{eqnarray*}
E[C^n_m(t)]&=&C^n_m(0)+n\mu\int_0^t E[C_s(u)C^{n-1}_m(u)]du-nd^{\prime}_2\int_0^t E[C^{n}_m(u)]du-n\gamma \int_0^t E[C^{n}_m(u)T(u)]du\\
&&+~\frac{1}{2}n(n-1)\bigg[\mu^2 \int_0^t E[C_s(u)C^{n-2}_m(u)]du+{d^{\prime}_2}^2 \int_0^t E[C^{n-1}_m(u)]du+\gamma^2 \int_0^t E[C^{n-1}_m(u)T(u)]du\bigg]
\end{eqnarray*}
After differentiation with respect to $t$, the above equation becomes
\begin{eqnarray*}
\frac{dE[C^n_m(t)]}{dt}&=&n\mu E[C_s(t)C^{n-1}_m(t)]-nd^{\prime}_2 E[C^{n}_m(t)]-n\gamma E[C^{n}_m(t)T(t)]\\
&&+~\frac{1}{2}n(n-1)\Big[\mu^2 E[C_s(t)C^{n-2}_m(t)]+{d^{\prime}_2}^2 E[C^{n-1}_m(t)]+\gamma^2 E[C^{n-1}_m(t)T(t)]\Big]
\end{eqnarray*}
From inequality \eqref{CML_C_s_Bound}, there exists $L>0$ such that
\begin{eqnarray}
E[C_s(t)C^{n-1}_m(t)]\le L E[C^{n-1}_m(t)]~~\text{and}~~E[C_s(t)C^{n-2}_m(t)]\le L E[C^{n-1}_m(t)].
\end{eqnarray}
Therefore,
\begin{eqnarray*}
\frac{dE[C^n_m(t)]}{dt}&\le&n\mu L E[C^{n-1}_m(t)]-nd^{\prime}_2 E[C^{n}_m(t)]-n\gamma E[C^{n}_m(t)T(t)]\\
&&+~\frac{1}{2}n(n-1)\Big[\mu^2L E[C^{n-2}_m(t)]+{d^{\prime}_2}^2 E[C^{n-1}_m(t)]+\gamma^2 E[C^{n-1}_m(t)T(t)]\Big]
\end{eqnarray*}
Using H\"older's inequality, we obtain
\begin{eqnarray*}
\frac{dE[C^n_m(t)]}{dt}&\le&n\mu L \Big\{E[C^{n}_m(t)]\Big\}^{\frac{n-1}{n}}-nd^{\prime}_2 E[C^{n}_m(t)]-n\gamma E[C^{n}_m(t)T(t)]\\
&&+~\frac{1}{2}n(n-1)\bigg[\mu^2L \Big\{E[C^{n}_m(t)]\Big\}^{\frac{n-2}{n}}+{d^{\prime}_2}^2 \Big\{E[C^{n}_m(t)]\Big\}^{\frac{n-1}{n}}+\gamma^2 E[C^{n-1}_m(t)T(t)]\bigg]
\end{eqnarray*}
Letting $E[C^{n}_m(t)]\ge 1$ leads to
\begin{eqnarray*}
\frac{dE[C^n_m(t)]}{dt}&\le&-nd^{\prime}_2 E[C^{n}_m(t)]+\frac{n}{2}\Big[2\mu L+(n-1)({d^{\prime}_2}^2+\mu^2L)\Big] \Big\{E[C^{n}_m(t)]\Big\}^{\frac{n-1}{n}}\\
&&-~\frac{n\gamma}{2} \Big[2+\gamma-n\gamma\Big] E[C^{n}_m(t)T(t)]
\end{eqnarray*}
For $\displaystyle{n\le 1+\frac{2}{\gamma}}$, we obtain
\begin{eqnarray*}
\frac{dE[C^n_m(t)]}{dt}&\le&-nd^{\prime}_2 E[C^{n}_m(t)]+\frac{n}{2}\Big[2\mu L+(n-1)({d^{\prime}_2}^2+\mu^2L)\Big] \Big\{E[C^{n}_m(t)]\Big\}^{\frac{n-1}{n}}
\end{eqnarray*}
Therefore
\begin{eqnarray}
\label{CML_C_m_Bound}
\limsup_{t\to\infty} E[C^n_m(t)]\le \left(\frac{A_2}{d^{\prime}_2}\right)^{n}:=M_2(n),
\end{eqnarray}
where
\begin{equation}
A_2 = \frac{1}{2}\Big[2\mu L+(n-1)({d^{\prime}_2}^2+\mu^2L)\Big].
\end{equation}
Finally, let $U_3(T(t))={[T(t)]}^n,~n\ge 1$. Using It\^o formula, we obtain from \eqref{CML_SDE_Model_Equations},
\begin{eqnarray*}
dU_3(T(t))&=&nT^{n-1}(t)\Big[\alpha C_m(t)T(t)-d^{\prime}_3T(t)\Big]dt+\frac{1}{2}n(n-1)T^{n-2}(t)\Big[h^2_{32}(t)+h^2_{33}(t)\Big]dt\\
&&+~nT^{n-1}(t)\Big[h_{32}(t)dW_2(t)+h_{33}(t)dW_3(t)\Big]
\end{eqnarray*}
Proceeding on similar lines, we obtain
\begin{eqnarray*}
\frac{dE[T^n(t)]}{dt}&=&n\alpha E[C_m(t)T^{n}(t)]-nd^{\prime}_3E[T^{n}(t)]
+\frac{1}{2}n(n-1)\Big[{d^{\prime}_3}^2 E[T^{n-1}(t)]+\alpha^2 E[C_m(t)T^{n-1}(t)]\Big]
\end{eqnarray*}
From inequality \eqref{CML_C_m_Bound}, there exists $\bar{L}>0$ such that
\[E[C_m(t)T^{n-1}(t)]\le \bar{L} E[T^{n-1}(t)],~~E[C_m(t)T^{n}(t)]\le \bar{L} E[T^{n}(t)].\]
Therefore,
\begin{eqnarray*}
\frac{dE[T^n(t)]}{dt}&\le&n\alpha \bar{L} E[T^{n}(t)]-nd^{\prime}_3E[T^{n}(t)]
+\frac{1}{2}n(n-1)\Big[{d^{\prime}_3}^2 E[T^{n-1}(t)]+\alpha^2\bar{L} E[T^{n-1}(t)]\Big]
\end{eqnarray*}
Using H\"older's inequality, we obtain
\begin{eqnarray*}
\frac{dE[T^n(t)]}{dt}&\le&-n\left(d^{\prime}_3-\alpha \bar{L}\right) E[T^{n}(t)]
+\frac{1}{2}n(n-1)({d^{\prime}_3}^2+\alpha^2\bar{L})\Big\{E[T^{n}(t)\Big\}^{\frac{n-1}{n}}
\end{eqnarray*}
Assuming $d^{\prime}_3>\alpha \bar{L}$, it follows that
\begin{eqnarray}
\label{CML_T_Bound}
\limsup_{t\to\infty} E[T^n(t)]\le \left(\frac{A_3}{d^{\prime}_3-\alpha \bar{L}}\right)^{n}:=M_3(n),
\end{eqnarray}
where
\begin{equation}
A_3 = \frac{1}{2}(n-1)({d^{\prime}_3}^2+\alpha^2\bar{L}).
\end{equation}
Hence the theorem follows.
\end{proof}

\begin{theorem}
\label{CML_Theorem_SDE_Stoch_Bounded}
Let $\displaystyle{X(t)=\left(C_s(t),C_m(t),T(t)\right)\in \mathbb{R}^3_+}$ be a solution of the stochastic system \eqref{CML_SDE_Model_Equations} with an initial condition $\displaystyle{X_0=(C_s(0),C_m(0),T(0)) \in \mathbb{R}^3_+}$. Then for every $\epsilon>0$, there is a constant $L$ for which
\[\limsup_{t\to\infty}P\Big\{||X(t)||\le L\Big\}\ge 1-\epsilon,\]
where $||X(t)||=\sqrt{C^2_s(t)+C^2_m(t)+T^2(t)}$. That is, the solution $X(t)$ is stochastically bounded (or bounded in probability).
\end{theorem}	
\begin{proof}
From Theorem \ref{CML_Theorem_SDE_Moment_Bounded}, it follows that
\[\limsup_{t\to\infty}E[C^n_s(t)+C^n_m(t)+T^n(t)]\le \sum_{i=1}^{3}M_i(n),~n\ge1.\]
Putting $n=2$, we obtain
\[\limsup_{t\to\infty}E\left[||X(t)||^2\right]=\limsup_{t\to\infty}E[C^2_s(t)+C^2_m(t)+T^2(t)]\le \sum_{i=1}^{3}M_i(2).\]
This implies that for $M>0$, there is a $t^{\prime}>0$ such that
\[E\left[||X(t)||^2\right]\le M,~\forall~t>t^{\prime}.\]
Since $E\left[||X(t)||^2\right]$ is continuous function of $t$, there exists $\bar{M}> 0$ such that
\[E\left[||X(t)||^2\right] \le \bar{M},~\forall~t\in [0,t^{\prime}].\]
Assuming $B=\max\{{M},~\bar{M}\}$, we obtain
\[E\left[||X(t)||^2\right] \le \bar{B},~\forall~t\in [0,\infty).\]
Hence
\[\sup_{t\ge 0} E\left[||X(t)||^2\right] \le \bar{B},~\forall~t\in [0,\infty).\]
Using Tchebychev's inequality, it follows for $L>0$ that
\[P\Big\{||X(t)||> L\Big\} \le \frac{\sup_{t \ge 0} E\left[||X(t)||^2\right]}{L^2} \le \frac{\bar{B}}{L^2},\]
which implies
\[ P\Big\{||X(t)||\le L\Big\} = 1-P\Big\{||X(t)||> L\Big\} \ge 1-\frac{\bar{B}}{L^2}.\]
Taking $\epsilon=\frac{\bar{B}}{L^2}$, we obtain the following result
\[\limsup_{t\to\infty}P\Big\{||X(t)||\le L\Big\}\ge 1-\epsilon.\]
\end{proof}

\begin{theorem}
\label{CML_Theorem_SDE_Sol_Exp_Stable}
Suppose the following conditions are satisfied:
\[C_1>0,~~C_2>0,~~C_3>0~~\text{and}~~C_1C_2>C_3,\]
where
\begin{eqnarray*}
C_1&=&2({d^{\prime}_2}+{d^{\prime}_3})+r^2K^{-2}+\eta-2r,\\
C_2&=&2({d^{\prime}_2}+{d^{\prime}_3})(r^2K^{-2}+\eta)-(r-d^{\prime}_2)^2-(r-d^{\prime}_3)^2-(d^{\prime}_2-d^{\prime}_3)^2-4r(d^{\prime}_2+d^{\prime}_3),\\
C_3&=&2r(d^{\prime}_2-d^{\prime}_3)^2
+(2d^{\prime}_2+d^{\prime}_3)(r-d^{\prime}_2)(r-d^{\prime}_3)\\
&&-~\left(r^2K^{-2}+\eta\right)(d^{\prime}_2-d^{\prime}_3)^2-2d^{\prime}_2(r-d^{\prime}_3)^2-2d^{\prime}_3(r-d^{\prime}_2)^2,
\end{eqnarray*}
with $\eta=\min\{r^2,~2rK^{-1}\}$. Then the disease free equilibrium (trivial solution) of the stochastic system \eqref{CML_SDE_Model_Equations} is almost surely exponentially stable.
\end{theorem}
\begin{proof}
Let \[V(X(t))=\ln\left[\alpha\mu C_s(t)+\alpha d^{\prime}_1C_m(t)+\gamma d^{\prime}_1 T(t)\right].\]
Using It\^o formula, we obtain
\begin{eqnarray*}
dV(X(t))&=&\Bigg[\frac{\alpha\mu rC_s(t)\left(1-{K}^{-1}C_s(t)\right)-\alpha d^{\prime}_1d^{\prime}_2C_m(t)-\gamma d^{\prime}_1d^{\prime}_3T(t)}{\alpha\mu C_s(t)+\alpha d^{\prime}_1C_m(t)+\gamma d^{\prime}_1 T(t)}\\
&&-~\frac{\alpha^2\mu^2r^2C_s(t)+\alpha^2\mu^2r^2K^{-2}C^2_s(t)+\alpha^2{d^{\prime}_1}^2{d^{\prime}_2}^2C_m(t)+\gamma^2{d^{\prime}_1}^2{d^{\prime}_3}^2 T(t)}{2\big[\alpha\mu C_s(t)+\alpha d^{\prime}_1C_m(t)+\gamma d^{\prime}_1 T(t)\big]^2}\Bigg]dt\\
&&+~\frac{\alpha(\mu h_{11}+d_1h_{21})dW_1(t)+d_1(\alpha h_{22}+\gamma h_{32})dW_2(t)+\gamma d_1h_{33}dW_3(t)}{\alpha\mu C_s(t)+\alpha d^{\prime}_1C_m(t)+\gamma d^{\prime}_1 T(t)}\\
&=&\frac{1}{2}\Big(\alpha\mu C_s(t)+\alpha d^{\prime}_1C_m(t)+\gamma d^{\prime}_1 T(t)\Big)^{-2}\Big[\alpha^2\mu^2\left(2r-r^2K^{-2}\right)C^2_s(t)-2\alpha^2{d^{\prime}_1}^2{d^{\prime}_2}^2C^2_m(t)\\
&&-~2\gamma^2{d^{\prime}_1}^2{d^{\prime}_3}^2T^2(t)-2\alpha^2\mu^2rK^{-1}C^3_s(t)-2\alpha^2\mu{d^{\prime}_1}({d^{\prime}_2}-r)C_s(t)C_m(t)\\
&&-~2\alpha\mu\gamma {d^{\prime}_1}({d^{\prime}_3}-r)C_s(t)T(t)-2\alpha\gamma {d^{\prime}_1}^2({d^{\prime}_2}+d^{\prime}_3)C_m(t)T(t)-2\alpha^2\mu d^{\prime}_1rK^{-1}C^2_s(t)C_m(t)\\
&&-~2\alpha\mu\gamma d^{\prime}_1rK^{-1}C^2_s(t)T(t)-\alpha^2\mu^2r^2C_s(t)-\alpha^2{d^{\prime}_1}^2{d^{\prime}_2}^2C_m(t)-\gamma^2{d^{\prime}_1}^2{d^{\prime}_3}^2 T(t)\Big]dt\\
&&+~\frac{\alpha(\mu h_{11}+d_1h_{21})dW_1(t)+d_1(\alpha h_{22}+\gamma h_{32})dW_2(t)+\gamma d_1h_{33}dW_3(t)}{\alpha\mu C_s(t)+\alpha d^{\prime}_1C_m(t)+\gamma d^{\prime}_1 T(t)}\\
&\le&\frac{1}{2}\Big(\alpha\mu C_s(t)+\alpha d^{\prime}_1C_m(t)+\gamma d^{\prime}_1 T(t)\Big)^{-2}\Big[\alpha^2\mu^2\left(2r-r^2K^{-2}-\eta \right)C^2_s(t)\\
&&-~2\alpha^2{d^{\prime}_1}^2{d^{\prime}_2}^2C^2_m(t)-2\gamma^2{d^{\prime}_1}^2{d^{\prime}_3}^2T^2(t)
-2\alpha^2\mu{d^{\prime}_1}({d^{\prime}_2}-r)C_s(t)C_m(t)\\
&&-~2\alpha\mu\gamma {d^{\prime}_1}({d^{\prime}_3}-r)C_s(t)T(t)-2\alpha\gamma {d^{\prime}_1}^2({d^{\prime}_2}+d^{\prime}_3)C_m(t)T(t)\Big]\\
&&+~\frac{\alpha(\mu h_{11}+d_1h_{21})dW_1(t)+d_1(\alpha h_{22}+\gamma h_{32})dW_2(t)+\gamma d_1h_{33}dW_3(t)}{\alpha\mu C_s(t)+\alpha d^{\prime}_1C_m(t)+\gamma d^{\prime}_1 T(t)},
\end{eqnarray*}	
where $\eta=\min\{r^2,~2rK^{-1}\}$.  Note that the following expression
\begin{eqnarray*}	
&&\alpha^2\mu^2\left(2r-r^2K^{-2}-\eta \right)C^2_s(t)-2\alpha^2{d^{\prime}_1}^2{d^{\prime}_2}^2C^2_m(t)
-2\gamma^2{d^{\prime}_1}^2{d^{\prime}_3}^2T^2(t)\\
&&-~2\alpha^2\mu{d^{\prime}_1}({d^{\prime}_2}-r)C_s(t)C_m(t)-2\alpha\mu\gamma {d^{\prime}_1}({d^{\prime}_3}-r)C_s(t)T(t)-2\alpha\gamma {d^{\prime}_1}^2({d^{\prime}_2}+d^{\prime}_3)C_m(t)T(t)
\end{eqnarray*}	
can be expressed as $Y^{\top}PY$, where
\[Y=
\begin{pmatrix}
\alpha\mu C_s(t) \\ \alpha d^{\prime}_1C_m(t) \\ \gamma d^{\prime}_1 T(t)	
\end{pmatrix}~~\text{and}~~	
P=
\begin{pmatrix}
2r-r^2K^{-2}-\eta & r-d^{\prime}_2 & r-d^{\prime}_3\\
r-d^{\prime}_2 & -2d^{\prime}_2 & -d^{\prime}_2-d^{\prime}_3\\
r-d^{\prime}_3 & -d^{\prime}_2-d^{\prime}_3 & -2d^{\prime}_3
\end{pmatrix}.		
\]	
The matrix $P$ is negative definite under the following conditions being satisfied:
\begin{equation}
C_1>0,~~C_2>0,~~C_3>0~~\text{and}~~C_1C_2>C_3,
\end{equation}
where
\begin{eqnarray*}
C_1&=&2({d^{\prime}_2}+{d^{\prime}_3})+r^2K^{-2}+\eta-2r,\\
C_2&=&2({d^{\prime}_2}+{d^{\prime}_3})(r^2K^{-2}+\eta)-(r-d^{\prime}_2)^2-(r-d^{\prime}_3)^2-(d^{\prime}_2-d^{\prime}_3)^2-4r(d^{\prime}_2+d^{\prime}_3),\\
C_3&=&2r(d^{\prime}_2-d^{\prime}_3)^2
+(2d^{\prime}_2+d^{\prime}_3)(r-d^{\prime}_2)(r-d^{\prime}_3)\\
&&-~\left(r^2K^{-2}+\eta\right)(d^{\prime}_2-d^{\prime}_3)^2-2d^{\prime}_2(r-d^{\prime}_3)^2-2d^{\prime}_3(r-d^{\prime}_2)^2.
\end{eqnarray*}	
Let $\lambda_{\text{max}}$ be the largest (negative) eigenvalue of $P$. Then
\[Y^{\top}PY\le -|\lambda_{\text{max}}|\left(\alpha^2\mu^2 C^2_s(t) + \alpha^2 {d^{\prime}_1}^2C^2_m(t) + \gamma^2 {d^{\prime}_1}^2 T^2(t)\right).\]
Therefore
\begin{eqnarray*}	
dV(X(t))&\le&-|\lambda_{\text{max}}|\frac{\left(\alpha^2\mu^2 C^2_s(t) + \alpha^2 {d^{\prime}_1}^2C^2_m(t) + \gamma^2 {d^{\prime}_1}^2 T^2(t)\right)}{2\left(\alpha\mu C_s(t)+\alpha d^{\prime}_1C_m(t)+\gamma d^{\prime}_1 T(t)\right)^{2}}dt\\
&&+~\frac{\alpha(\mu h_{11}+d_1h_{21})dW_1(t)+d_1(\alpha h_{22}+\gamma h_{32})dW_2(t)+\gamma d_1h_{33}dW_3(t)}{\alpha\mu C_s(t)+\alpha d^{\prime}_1C_m(t)+\gamma d^{\prime}_1 T(t)}
\end{eqnarray*}	
Note that
\begin{eqnarray}	
\label{Inequality_Square}
\left(\alpha\mu C_s(t)+\alpha d^{\prime}_1C_m(t)+\gamma d^{\prime}_1 T(t)\right)^{2} &\le& 3\left(\alpha^2\mu^2 C^2_s(t) + \alpha^2 {d^{\prime}_1}^2C^2_m(t) + \gamma^2 {d^{\prime}_1}^2 T^2(t)\right)
\end{eqnarray}	
Using inequality \eqref{Inequality_Square}, we obtain
\begin{eqnarray*}	
dV(X(t))&\le&-\frac{1}{6}|\lambda_{\text{max}}|dt+\frac{\alpha(\mu h_{11}+d_1h_{21})dW_1(t)+d_1(\alpha h_{22}+\gamma h_{32})dW_2(t)+\gamma d_1h_{33}dW_3(t)}{\alpha\mu C_s(t)+\alpha d^{\prime}_1C_m(t)+\gamma d^{\prime}_1 T(t)}
\end{eqnarray*}	
Integrating and using the fact that
\[\limsup_{t\to\infty}\frac{1}{t}\left|W_i(t)\right|=0,~~i=1,2,3,\]
we obtain
\[\limsup_{t\to\infty}\frac{1}{t}V(X(t))\le -\frac{1}{6}|\lambda_{\text{max}}|<0~~\text{a.s.}\]
\textit{i.e.,}\[\limsup_{t\to\infty}\frac{1}{t}\ln\left[\alpha\mu C_s(t)+\alpha d^{\prime}_1C_m(t)+\gamma d^{\prime}_1 T(t)\right]<0~~\text{a.s.}\]
Hence $C_s(t)\to 0$, $C_m(t)\to 0$ and $T(t)\to 0$ a.s. as $t\to\infty.$ This completes the proof of Theorem \ref{CML_Theorem_SDE_Sol_Exp_Stable}.
\end{proof}

We now determine the probability of CML extinction in the disease persistent stage. In order to ascertain this probability, we identify the possible state changes during the time interval $\Delta t$ tabulated in Table \ref{CML_Table_State_Changes}.
\begin{table}[ht]
\centering
{\begin{tabular}{llll}
\hline
i & State change $((\Delta X)_{i})$ & Probability $(p_{i})$ & Description \\
\hline
1 & $(1,0,0)^{\top}$ & $rC_{s}\left(1-\frac{C_{s}}{K}\right)\Delta t$ & Source of a CML stem cell\\
2 & $(-1,0,0)^{\top}$ & $d_{1}^{\prime}C_{s}\Delta t$ &  Neutralizing by Wnt or natural death of a CML stem cell\\
3 & $(0,1,0)^{\top}$ & $\mu C_{s}\Delta t$ & Production of a mature CML cell\\
4 & $(0,-1,0)^{\top}$ & $(\mu T+d_2^{\prime})C_{m}\Delta t$ & Neutralizing by CTL/ TKI or natural death of a mature CML cell\\
5 & $(0,0,1)^{\top}$ & $\alpha C_{m}T\Delta t$ & Development of a CTL cell\\
6 & $(0,0,-1)^{\top}$ & $d_3^{\prime}T\Delta t$ & Stimulating by IFN-$\alpha$ along with natural death of a CTL cell\\
\hline
\end{tabular}}
\caption{Possible state changes during $\Delta t$.}
\label{CML_Table_State_Changes}
\end{table}

We derive the offspring probability generating functions (PGFs) using the theory of multitype continuous-time branching process approximation \cite{Kimmel02,Vidurupola14}:
\begin{enumerate}
\item The offspring PGF for $C_{s}$ (corresponding to $i=1,2,3$ in Table \ref{CML_Table_State_Changes}) is
\[f_1(s_1,s_2)=\frac{rs_1^2+d_1^{\prime}+\mu s_1s_2}{r+d_1^{\prime}+\mu}.\]
\item The offspring PGF for $C_{m}$ (corresponding to $i=4,5$ in Table \ref{CML_Table_State_Changes}) is
\[f_2(s_1,s_2)=\frac{d_2^{\prime}}{d_2^{\prime}}=1.\]
\end{enumerate}
Consequently, the expectation matrix for the PGFs is
\[M=
\begin{bmatrix}
\frac{2r+\mu}{r+d_1^{\prime}+\mu} & 0\\
\frac{\mu}{r+d_1^{\prime}+\mu} & 0
\end{bmatrix}.\]
If $(q_1,q_2)\in (0,1)\times (0,1)$ is a fixed point, then
\[q_1=\frac{rq_1^2+d_1^{\prime}+\mu q_1q_2}{r+d_1^{\prime}+\mu}~\text{and}~q_2=1.\]
This implies $\displaystyle{q_1=1,~\frac{d_1^{\prime}}{r}}$ and $\displaystyle{q_2=1}$. For the disease persistent stage, $r>d_{1}^{\prime}$, we consider $\displaystyle{q_1=\frac{d_1^{\prime}}{r}}$.
Using Galton-Watson branching theory \cite{Harris63} , it follows that the probability of CML extinction is,
\[q=q_{1}^{m}=\left(\frac{d_1^{\prime}}{r}\right)^m,~\text{where}~m=C_{s}(0).\]

\section{Numerical Results}
\label{CML_Numerical_Results}

In this section, we dwell upon the numerical illustration of the results obtained for both the deterministic and the stochastic CML models. The computations were run using MATLAB\textsuperscript{\textregistered} for a set of parameter values given in Table \ref{CML_Table_Parameter_Values}, except for the cases with some particular parameter values (which will be specified in the discussion). Accordingly, we take the initial condition to be $\displaystyle{\left(C_{s}(0),C_{m}(0),T(0)\right)=\left(20,500,1000\right)}$, for the deterministic as well as the stochastic model simulations, with the number of sample paths for the latter (generated using the Euler-Maruyama method \cite{Allen10}) being $5000$. In order to numerically show the stability of the equilibria, we choose, two different values of the parameter $\mu$, namely, $\mu=1$ and $\mu=10$, to generate various scenarios. The set of all other parameter values chosen from Table \ref{CML_Table_Parameter_Values} with $\mu=1$ results in $P^{*}=0.3333(<1)$ which gives the existence of the equilibrium $\displaystyle{\bar{E}}=\left(38.0211,353.6850,0\right)$.
Similarly, for $\mu=10$ we get $P^{*}=3.3335 (>1)$ which gives the existence of the equilibrium $\displaystyle{E^{*}=\left(38.0211,1060.9756,68785.3432\right)}$.
\begin{table}[ht]
\centering
{\begin{tabular}{lll}
\hline
Parameter & Value [Reference] & Unit \\
\hline
$r$ & $1$ & \text{day}$^{-1}$\\
$K$ & $41.667$ \cite{Besse18} & \text{cells}~\text{ml}$^{-1}$\\
$\gamma$ & $3.647\times 10^{-6}$ \cite{Besse18} & \text{ml}~\text{cell}$^{-1}$~\text{day}$^{-1}$\\
$\mu $ & $10$ & \text{day}$^{-1}$\\
$d_{1}$ & $0.0375$ \cite{Besse18} & \text{day}$^{-1}$\\
$d_{2}$ & $0.0375$ \cite{Besse18} & \text{day}$^{-1}$\\
$d_{3}$ & $0.5$ \cite{Berezansky12} & \text{day}$^{-1}$\\
$\alpha$ & $0.00041$ \cite{Berezansky12} & \text{ml}~\text{cell}$^{-1}$~\text{day}$^{-1}$\\
$\beta_{1}$ & $0.05$ & \text{day}$^{-1}$\\
$\beta_{2}$ & $0.07$ \cite{Berezansky12} & \text{day}$^{-1}$\\
$\beta_{3}$ & $0.065$ \cite{Berezansky12} & \text{day}$^{-1}$\\
\hline
\end{tabular}}
\caption{Parameter values for numerical results}
\label{CML_Table_Parameter_Values}
\end{table}

We first examine the dynamics for the case $P^{*}<1$ corresponding to the case when $\bar{E}$ exists and is stable. The level of the CML stem cells increases rapidly, eventually stabilizing at the peak level of $C_{s}=38.0211$ shortly after the initial time (Figure \ref{CML_Local_Stability_Second_Equilibrium_A}). As for the mature CML cells, the level gradually decreases and eventually stabilizes at the minimum level of $C_{m}=353.6850$ (Figure \ref{CML_Local_Stability_Second_Equilibrium_B}). However, it takes more time to reach the corresponding stabilized level as compared to the case of CML stem cells. Further, the level of $T$ gradually decreases, eventually reaching zero level beyond which time there is cessation of the CTL immune response (Figure \ref{CML_Local_Stability_Second_Equilibrium_C}). We conclude that the reduction in the level of the mature CML cells, consequently, results in gradual decline with eventual non-response of the CTLs. Consequently, the patient stabilizes at the less severe disease equilibrium $\bar{E}$. We now examine the dynamics for the case $P^{*}>1$ when $E^{*}$ exists and is stable. Similar to the case of $P^{*}<1$, the population level of the CML stem cells increases very fast and in the long run stabilizes at the peak level of $C_{s}=38.0211$ (Figure \ref{CML_Local_Stability_Third_Equilibrium_A}). The level of the mature CML cells in this case increases very fast until the CTL immune response is stimulated to an efficient level. As a result of the affect of CTL immune response, the level of the mature CML cells keeps dropping. Consequently, due to less concentration of the mature CML cells, the generation of CTL immune response also decreases after elapse of some time and eventually the levels of the mature CML cells and the CTL immune response stabilize at the level of $C_{m}=1060.9756$ and $T=68785.3432$, respectively. These results are illustrated in Figure \ref{CML_Local_Stability_Third_Equilibrium_B} and Figure \ref{CML_Local_Stability_Third_Equilibrium_C}, respectively. The simulations for the case of $P^{*}>1$ suggests that sufficient level of mature CML cells results in continued CTL immune response, eventually sending the patient to the most severe disease equilibrium.

We now assess the pharmacokinetic impact resulting from administration of Wnt/$\beta$-catenin signaling ($\beta_{1}$), TKI ($\beta_{2}$) and IFN-$\alpha$ ($\beta_{3}$) for several scenarios. Firstly, we consider the impact of monotherapy of TKI, followed by dual therapy of TKI with  Wnt/$\beta$-catenin and IFN-$\alpha$, individually, with eventual consideration of all three. In order to present the effect of administration of TKI ($\beta_{2}$) in absence of Wnt/$\beta$-catenin signaling ($\beta_{1}$) and IFN-$\alpha$ ($\beta_{3}$), we ran the simulations
with $\beta_{1}=0$ and $\beta_{3}=0$, for five different values of $\beta_{2}$, namely, $0$, $0.05$, $0.1$, $0.5$ and $0.8$, having corresponding $P^{*}$ values of $8.7695$, $3.7584$, $2.3917$, $0.6118$ and $0.3927$. The simulations for CML stem cells, mature CML cells and the CTL immune response are presented in Figure \ref{CML_Effect_of_Beta_2}. From Figure \ref{CML_Effect_of_Beta_2_A}, it is observed that variation of $\beta_{2}$ has no impact at all on the dynamics of the CML stem cells. However, with an increase in the values of $\beta_{2}$, the values of the mature CML cells converge to ${C}_{m}^{*}$ for the first three values of $\beta_{2}$ (which corresponds to $P^{*}>1$), but converges to the lower value $\bar{C}_{m}$, with an increase for the other two values of $\beta_{2}$ (which corresponds to $P^{*}<1$) (Figure \ref{CML_Effect_of_Beta_2_B}). For the CTL immune response, one observes a decrease corresponding to increase in the values of $\beta_{2}$ (Figure \ref{CML_Effect_of_Beta_2_C}). For the last two values of $\beta_{2}$, the CTL immune response is eventually not sustained. From the simulations, we observe that the increase in $\beta_{2}$, upto a certain level does not affect the eventual stabilized level $E^{*}$ of the mature CML cells. Upon further increase of $\beta_{2}$, the impact is that the level of the mature CML cells decreases and eventually the progression stabilizes at a level $\bar{E}$. Thus, in this case, the administration of high dosage is essential for achieving desired therapeutic response in terms of reduction of mature CML cells, although there is no impact on CML stem cells at all. This result supports the observations made in \cite{Gallipoli11,Rea12}

The impact of administration of dual therapy of Wnt/$\beta$-catenin signaling ($\beta_{1}$) and TKI ($\beta_{2}$), in absence of IFN-$\alpha$ ($\beta_{3}$), are presented in Figure \ref{CML_Beta1_Beta2_Effect}. The set of $\left(\beta_{1},\beta_{2}\right)$ values chosen are $(0,0)$, $(0.5,0.05)$, $(0.1,0.2)$, $(0.1,0.5)$, $(0.3,0.5)$ and $(0.3,0.6)$, with the corresponding values of $P^{*}$ being $8.7695$, $1.8059$, $1.2408$, $0.5483$, $0.4211$ and $0.3551$. In this case, the values of CML stem cells show a decreasing trend resulting from the corresponding increase in the values of $\beta_{1}$ only (irrespective of $\beta_{2}$) (Figure \ref{CML_Beta1_Beta2_Effect_A}). For the values $(0.1,0.2)$, $(0.1,0.5)$ of $(\beta_{1},\beta_{2})$, the trajectories for CML stem cells converge along the same path as in the case where $\beta_{1}$ is left unchanged. A similar scenario is observed in case of $(0.3,0.5)$ and $(0.3,0.6)$. For the dynamics of mature CML cells, we observe that, for the case with $P^{*}>1$, there is an eventual convergence to a certain identical level of mature CML cells, while the convergent level gradually gets lowered with a gradual decrease in $P^{*}$ when $P^{*}<1$ (Figure \ref{CML_Beta1_Beta2_Effect_B}). The decrease in value of $P^{*}$ results in the corresponding decrease in the levels of CTL immune response and even converges to zero level for the cases where $P^{*}<1$ (Figure \ref{CML_Beta1_Beta2_Effect_C}). The simulation suggests that this dual therapy is not effective on mature CML cells in case of $P^{*}>1$. However, this therapy has an impact on CML stem cells, which was observed in \cite{Zhang13B,Agarwal17,Luis12}.

We now observe the pattern of dual therapy of TKI ($\beta_{2}$) and IFN-$\alpha$ ($\beta_{3}$), in absence of Wnt/$\beta$-catenin signaling ($\beta_{1}$), in Figure \ref{CML_Beta2_Beta3_Effect}. The set of $\left(\beta_{2},\beta_{3}\right)$ values chosen are $(0,0)$, $(0.1,0.1)$, $(0.3,0.1)$, $(0.3,0.2)$, $(0.8,0.05)$ and $(0.2,0.4)$, with the corresponding values of $P^{*}$ being $8.7695$, $2.9896$, $1.2179$, $1.6239$, $0.4363$ and $6.9233$. Interestingly, the absence of $\beta_{1}$ results in the dynamics of CML stem cells remaining unchanged irrespective of the various combinations of $\beta_{2}$ and $\beta_{3}$ being applicable (Figure \ref{CML_Beta2_Beta3_Effect_A}). From Figure \ref{CML_Beta2_Beta3_Effect_B}, it is observed for the values $(0.1,0.1)$ and $(0.3,0.1)$ of $\left(\beta_{2},\beta_{3}\right)$ that, the concentration of mature CML cells does not change in terms of its stabilized level due to unchanged IFN-$\alpha$ dosage. Further, in case of the value $(0.8,0.05)$ of $(\beta_{2},\beta_{3})$, the level of mature CML cells reaches a low level, in spite of $\beta_{3}$ being very less, since higher value of $\beta_{2}$ results in $P^{*}<1$. For the value $(0.2,0.4)$, we observe that the mature CML cell stabilizes at the least level (Figure \ref{CML_Beta2_Beta3_Effect_B}), as CTL response increases to very high level in this case (Figure \ref{CML_Beta2_Beta3_Effect_C}). This shows that a low dosage of $\beta_{2}$ can be more effective in combination with a relatively high dosage of $\beta_{3}$. The simulations for the values $(0.3,0.1)$ and $(0.3,0.2)$ show that the level of mature CML cells gets reduced (Figure \ref{CML_Beta2_Beta3_Effect_B}), due to increase in dosage of IFN-$\alpha$, resulting in stimulation of CTL immune response (Figure \ref{CML_Beta2_Beta3_Effect_C}), a phenomenon which was not observed in case of monotherapy of TKI. Thus, the increase in IFN-$\alpha$ results in stimulation of CTL immune response, which is responsible for eventual decrease in the level of mature CML cells. However, this dual therapy does not have any effect on the CML stem cell population.

Finally, we analyze the clinical implications of the combination therapy of Wnt/$\beta$-catenin signaling ($\beta_{1}$), TKI ($\beta_{2}$) and IFN-$\alpha$ ($\beta_{3}$), as illustrated in Figure \ref{CML_Combination_Effect}. The sets of $\left(\beta_{1},\beta_{2},\beta_{3}\right)$ values are taken to be $(0,0,0)$, $(0.1,0.2,0)$, $(0,0.2,0.1)$, $(0.1,0.2,0.1)$, $(0.6,0.5,0.4)$ and $(0.3,0.4,0.1)$. Accordingly, the corresponding values of $P^{*}$ are $8.7695$, $1.2408$, $1.7308$, $1.5509$, $1.1521$ and $0.6467$. The pattern of eventual convergence of the CML stem cells consistently decreases with the corresponding increase in the value of Wnt/$\beta$-catenin signaling (Figure \ref{CML_Combination_Effect_A}). The sets $(0.1,0.2,0)$ and $(0.1,0.2,0.1)$ of $\left(\beta_{1},\beta_{2},\beta_{3}\right)$ show that a small addition of IFN-$\alpha$ results in a sufficient therapeutic effectiveness so as to achieve reduction in the levels of mature CML cells (Figure \ref{CML_Combination_Effect_B}), with an increase of CTL immune response (Figure \ref{CML_Combination_Effect_C}). The combinations of $(0,0.2,0.1)$ and $(0.1,0.2,0.1)$ shows that the change in Wnt/$\beta$-catenin signaling (without changing TKI or IFN-$\alpha$) does not have any effect on the stabilized level of mature CML cells, but the convergence to the stable level is achieved at different times for different values of Wnt/$\beta$-catenin signaling (Figure \ref{CML_Combination_Effect_B}). The administration of the targeted therapy of Wnt/$\beta$-catenin results in decreasing level of CML stem cells (Figure \ref{CML_Combination_Effect_A}) and consequently a decline in the level of CTL immune response is observed (Figure \ref{CML_Combination_Effect_C}). Thus the combination dosage $(0.1,0.2,0.1)$ has a significant impact on the level of both CML stem cells and mature CML cells (observed as compared to the dosage $(0,0,0)$). A choice of high combination dosage of $(0.6,0.5,0.4)$ results in an adequate therapeutic response in terms of reduction in the concentration of CML stem cells as well as mature CML cells. Lastly, in case of $(0.3,0.4,0.1)$, a relatively high amount of Wnt/$\beta$-catenin signaling and TKI in combination with less IFN-$\alpha$ corresponding to $P^{*}<1$ results in decrease being observed for both the CML stem cells as well the mature CML cells. Consequently, the CTL immune response eventually goes to zero. Thus, a suitable combination of Wnt/$\beta$-catenin signaling, TKI and IFN-$\alpha$ can be very effective in the suppression and reduction of both the types of CML cells as well as in the stimulation of CTL response.

We qualitatively summarize the numerical results depicting the dynamics corresponding to the effect of therapeutic protocols, in Table \ref{CML_Table_Summary}.
\begin{table}[ht]
{\scriptsize
\centering
{\begin{tabular}{lllll}
\hline
& $\beta_{2}\ne 0$ & $\beta_{1}\ne 0, \beta_{2}\ne 0$ & $\beta_{2}\ne 0, \beta_{3}\ne 0$ & $\beta_{1}\ne 0, \beta_{2}\ne 0, \beta_{3}\ne 0$ \\
\hline
$C_{s}$ & Unchanged & Decline & Unchanged & Decline \\
$C_{m}$ & Decline at high dosage & Decline at high dosage & Decline & Fast decline \\
$T$ & Decline & Decline & Dosage dependent decline/increase & Dosage dependent decline/increase\\
\hline
\end{tabular}}
\caption{Qualitative summary of the therapeutic results.}
\label{CML_Table_Summary}
}
\end{table}

We now perform the numerical simulation for the stochastic model \eqref{CML_SDE_Model_Equations} and present the histograms, in order to understand the probability distribution, at $t=50$ (chosen for illustrative purpose) based on $5000$ independent simulations, for the CML stem cells, mature CML cells and CTL immune response. An illustrative sample path along with the mean sample path, mean sample path minus standard deviation and mean sample path plus standard deviation is presented in Figures \ref{CML_SDE_A}, \ref{CML_SDE_B} and \ref{CML_SDE_C}, for the CML stem cells, mature CML cells and CTL immune response, respectively. The mean ($\mu_{X_i}$) and standard deviation ($\sigma_{X_i}$) at $t=50$, resulting from the $5000$ simulations are
\[\mu_{C_{s}}=37.6479,~\mu_{C_{m}}=1054.2897,~\mu_{T}=68961.1024\]
and
\[\sigma_{C_{s}}=4.3124,~\sigma_{C_{m}}=452.4636,~\sigma_{T}=37889.7108,\]
respectively. The standard deviation for the case of CML stem cells is relatively less than that of mature CML cells and CTL immune response, with respect to the population size. Moreover, the stochastic mean, $\left(\mu_{C_{s}},\mu_{C_{m}},\mu_{T}\right)$ is very close to its deterministic equilibrium value $E^{*}$. The stochastic fluctuation has a lesser impact on the dynamics of CML stem cells and it has a greater impact on CTL population.

Further, we present the probability histogram plots in Figure \ref{CML_Histogram}, and observe that the probability distribution of CML stem cells and mature CML cells fit the normal curve. The mean and variance sets used to plot the normal distribution in Figures \ref{CML_Hist_A}, \ref{CML_Hist_B} and \ref{CML_Hist_C} are $\left(\mu_{C_{s}},\sigma_{C_{s}}^{2}\right)$, $\left(\mu_{C_{m}},\sigma_{C_{m}}^{2}\right)$ and $\left(\mu_{T},\sigma_{T}^{2}\right)$, respectively.

Recall that the probability of CML extinction is dependent on both Wnt/$\beta$-catenin signaling ($\beta_{1}$) as well as the initial level of the CML stem cells ($C_{s}(0)$) which is depicted in Figure \ref{CML_Extinction}. Figure \ref{CML_Extinction} illustrates the probability of disease extinction depending on various dosages of Wnt/$\beta$-catenin signaling as well as various initial levels of the CML stem cells. It is observed from Figure \ref{CML_Extinction_Cs} that the probability of CML extinction is inversely proportional to the initial value of CML stem cells and it becomes very less (near to zero) if the initial CML stem cells exceeds $2$. Figure \ref{CML_Extinction_Beta1} suggests that the probability of disease extinction can be improved by increasing the dosage of Wnt/$\beta$-catenin signaling. It can be seen that the probability of CML extinction is greater when the initial level of the CML stem cells is less and Wnt/$\beta$-catenin signaling is high. The tabulated values of the parameters (Table \ref{CML_Table_Parameter_Values}) gives $P^{*}=3.26051$ and $q_{1}=0.1075$. Then the probability of CML extinction is $\displaystyle{q=\left(0.1075\right)^{C_{s}(0)}}$. As an illustration, we consider $C_{s}(0)=20$. Then $\displaystyle{q=\left(0.1075\right)^{C_{s}(0)}\approx 4.2478\times 10^{-20}}$, which shows that the probability of cure is almost non-existent and the prognosis is very poor. If we want the probability of CML extinction to be greater than a desired value (say $q_{e})$, then the dosage of the required Wnt/$\beta$-catenin signaling is as follows,
\[q=q_{1}^{C_{s}(0)}>q_{e}\Longrightarrow C_{s}(0)\ln \left(\frac{d_{1}+\beta_{1}}{r}\right)>\ln q_{e}\Longrightarrow \beta_{1}>r\exp\left({\frac{\ln q_{e}}{C_s(0)}}\right)-d_{1}.\]

\section{Conclusion}
\label{CML_Conclusion}

This work focuses on the modeling and quantitative analysis for the disease progression dynamics of CML in consideration with CTL immune response. The pharmacokinetics in presence of the traditional therapeutic protocol of IFN-$\alpha$ and TKI is extended to include Wnt/$\beta$-catenin signaling. The deterministic model is proposed and analyzed. The stability of the severe and less severe equilibria (in case of the natural progression rate of CML stem cells exceeding the sum of the rate of natural clearance and the efficacy of TKI dosage) is presented with subject to a condition on a threshold parameter $P^*$. However, the deterministic model does not provide any insight into the likelihood of remission. This is addressed by the introduction
of stochastic model which encapsulates this critical aspect of long term prognosis of CML remission and the probability of extinction of the disease. The well-posedness of both the deterministic model and stochastic model (in probability sense) is proved by showing the existence, uniqueness, boundedness and non-negativity of the solution. The stochastic mean solution is very close to its deterministic severe equilibrium value $E^*$. The stochastic noise is relatively more impactful on the population of mature CML cells and CTL cells. The stochastic model also showed the almost surely exponential stability of the disease free equilibrium under certain conditions. Further, the model predicts improved therapeutic response upon administration of Wnt/$\beta$-catenin signaling in addition to the IFN-$\alpha$ and TKI, particularly due to the resulting reduced side effects of the latter two. The results obtained demonstrates that the combination therapy is most successful in terms of the reduction of CML stem cells and mature CML cells with the stimulation of CTL immune response.

\clearpage

\begin{figure}[!ht]
\centering
\begin{subfigure}[b]{0.25\textwidth}
\includegraphics[width=\textwidth,height=4cm]{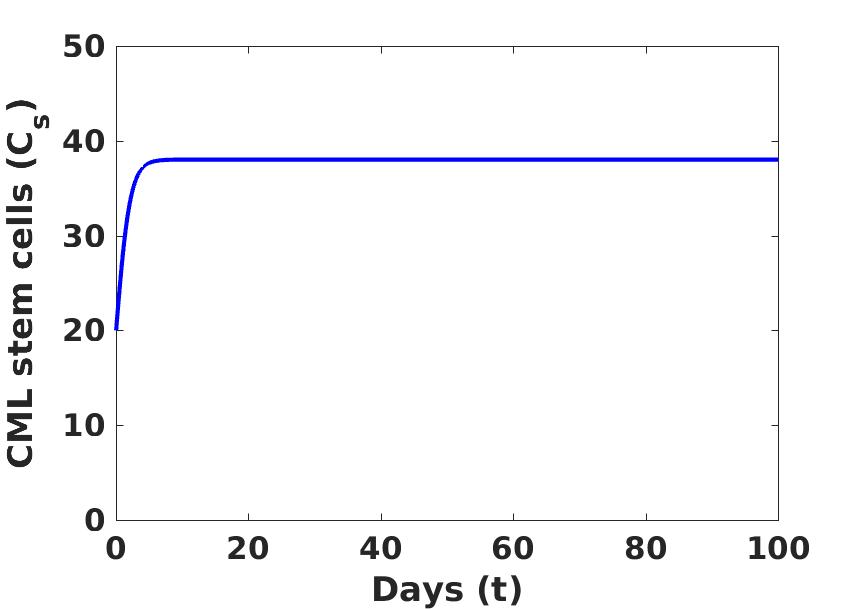}
\caption{}
\label{CML_Local_Stability_Second_Equilibrium_A}
\end{subfigure}
\quad
\begin{subfigure}[b]{0.25\textwidth}
\includegraphics[width=\textwidth,height=4cm]{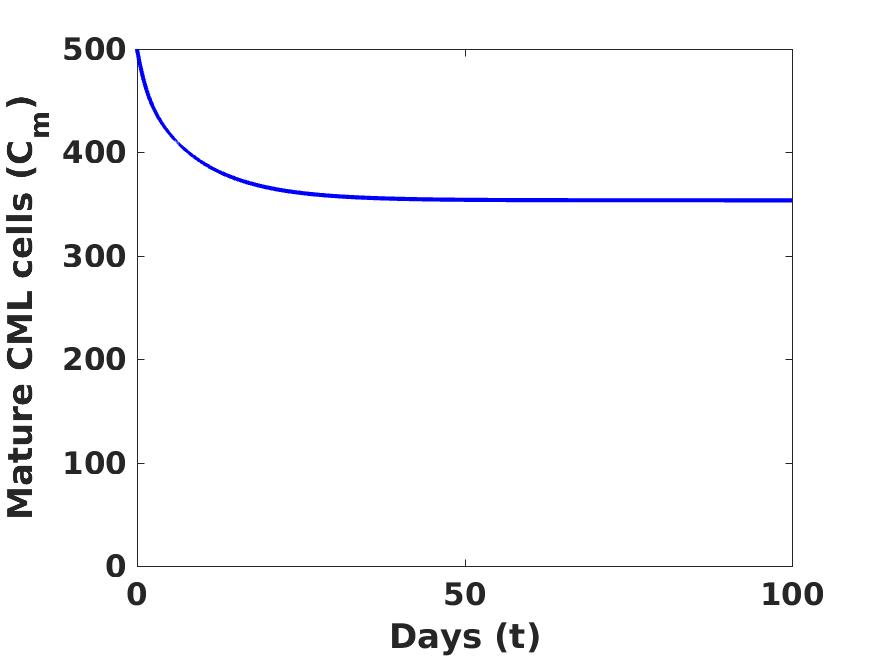}
\caption{}
\label{CML_Local_Stability_Second_Equilibrium_B}
\end{subfigure}
\quad
\begin{subfigure}[b]{0.25\textwidth}
\includegraphics[width=\textwidth,height=4cm]{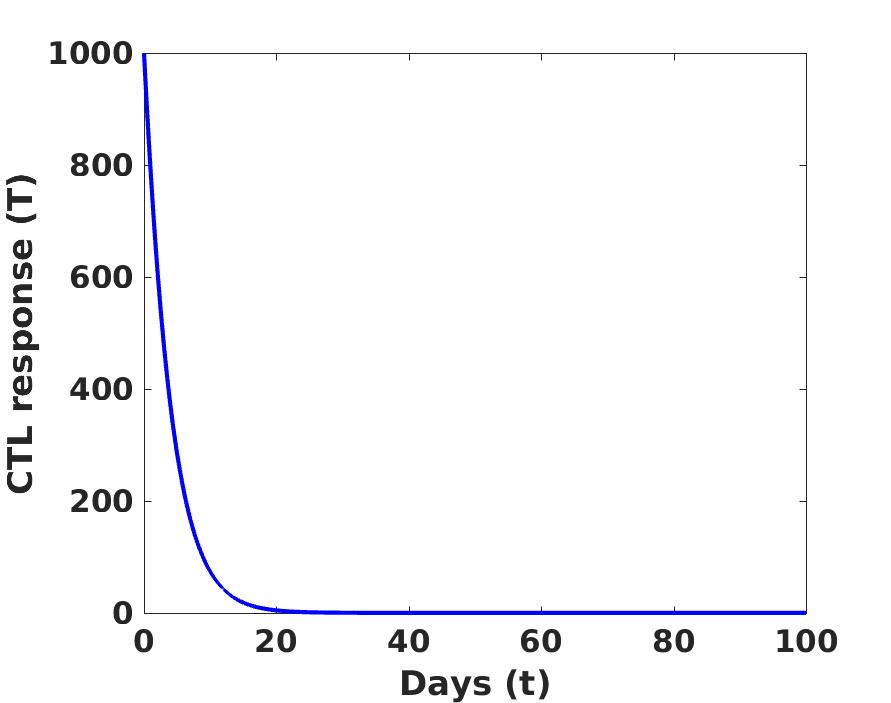}
\caption{}
\label{CML_Local_Stability_Second_Equilibrium_C}
\end{subfigure}
\caption{Dynamics of $C_{s}$, $C_{m}$ and $T$ for the case $P^{*}<1$.}
\label{CML_Local_Stability_Second_Equilibrium}
\end{figure}
\vspace{1cm}

\begin{figure}[!ht]
\centering
\begin{subfigure}[b]{0.25\textwidth}
\includegraphics[width=\textwidth,height=4cm]{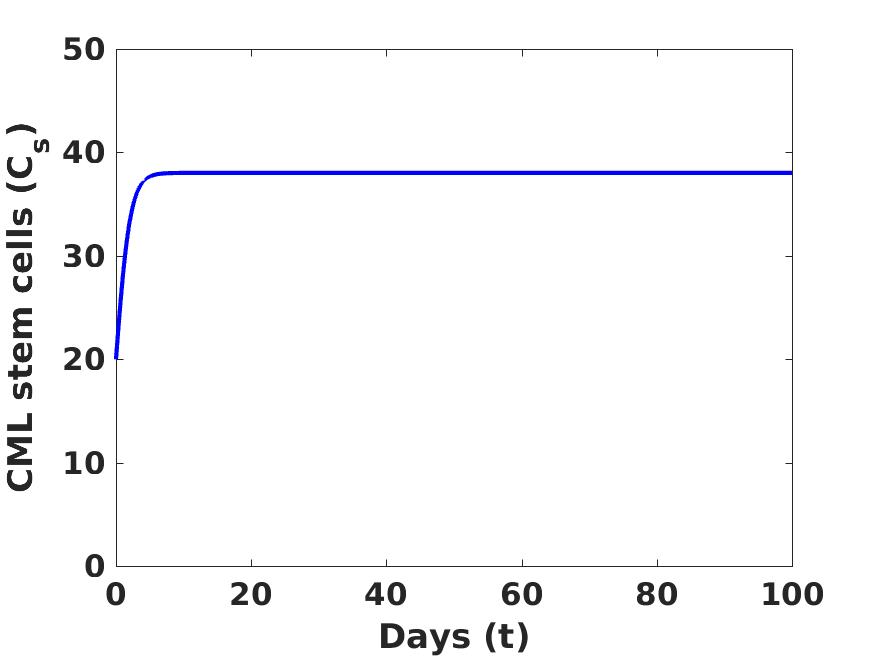}
\caption{}
\label{CML_Local_Stability_Third_Equilibrium_A}
\end{subfigure}
\quad
\begin{subfigure}[b]{0.25\textwidth}
\includegraphics[width=\textwidth,height=4cm]{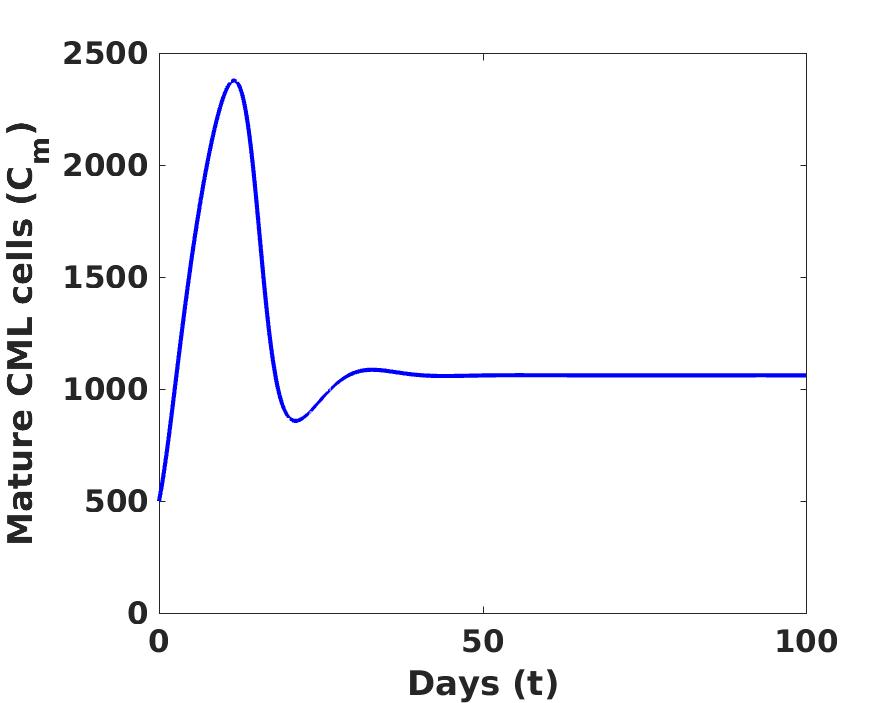}
\caption{}
\label{CML_Local_Stability_Third_Equilibrium_B}
\end{subfigure}
\quad
\begin{subfigure}[b]{0.25\textwidth}
\includegraphics[width=\textwidth,height=4cm]{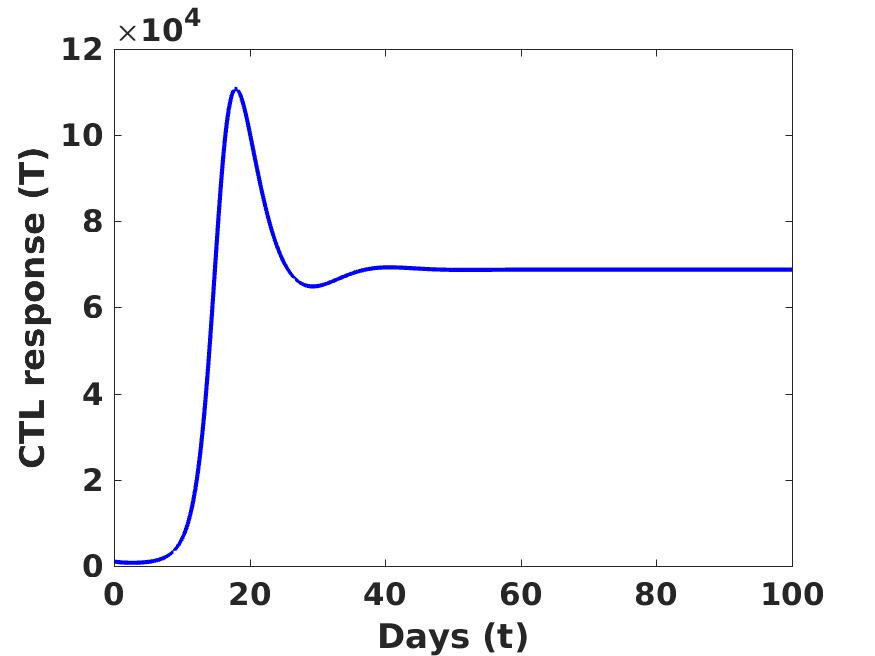}
\caption{}
\label{CML_Local_Stability_Third_Equilibrium_C}
\end{subfigure}
\caption{Dynamics of $C_{s}$, $C_{m}$ and $T$ for the case $P^{*}>1$.}
\label{CML_Local_Stability_Third_Equilibrium}
\end{figure}
\vspace{1cm}

\begin{figure}[!ht]
\centering
\begin{subfigure}[b]{0.25\textwidth}
\includegraphics[width=\textwidth,height=4cm]{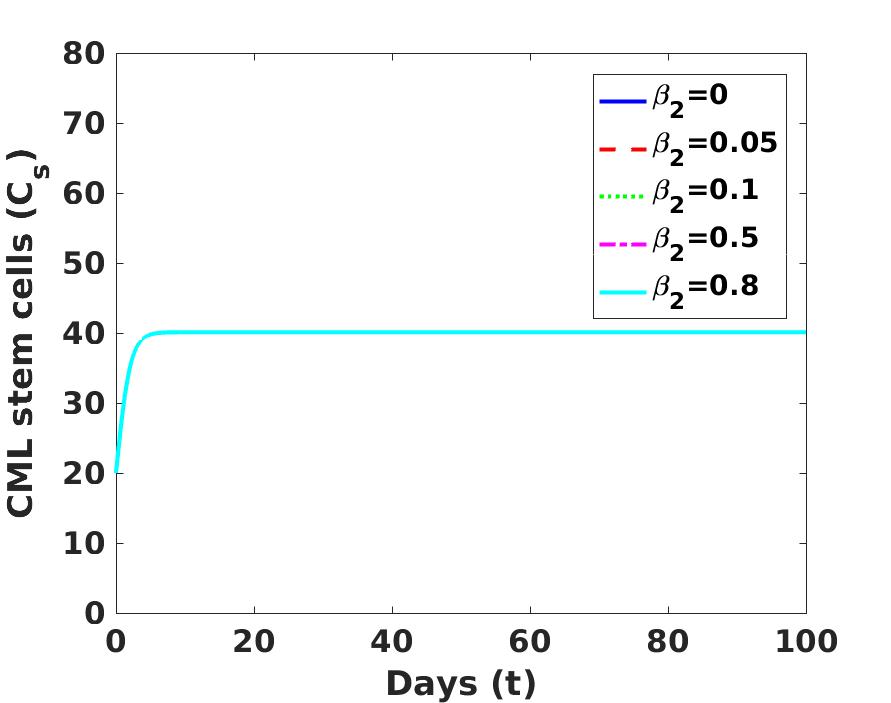}
\caption{}
\label{CML_Effect_of_Beta_2_A}
\end{subfigure}
\quad
\begin{subfigure}[b]{0.25\textwidth}
\includegraphics[width=\textwidth,height=4cm]{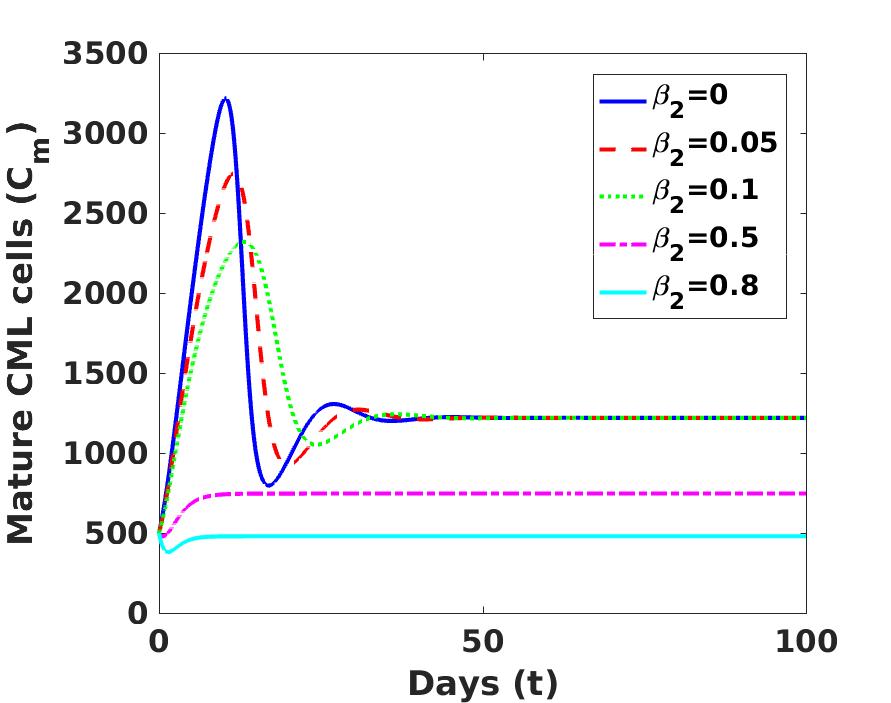}
\caption{}
\label{CML_Effect_of_Beta_2_B}
\end{subfigure}
\quad
\begin{subfigure}[b]{0.25\textwidth}
\includegraphics[width=\textwidth,height=4cm]{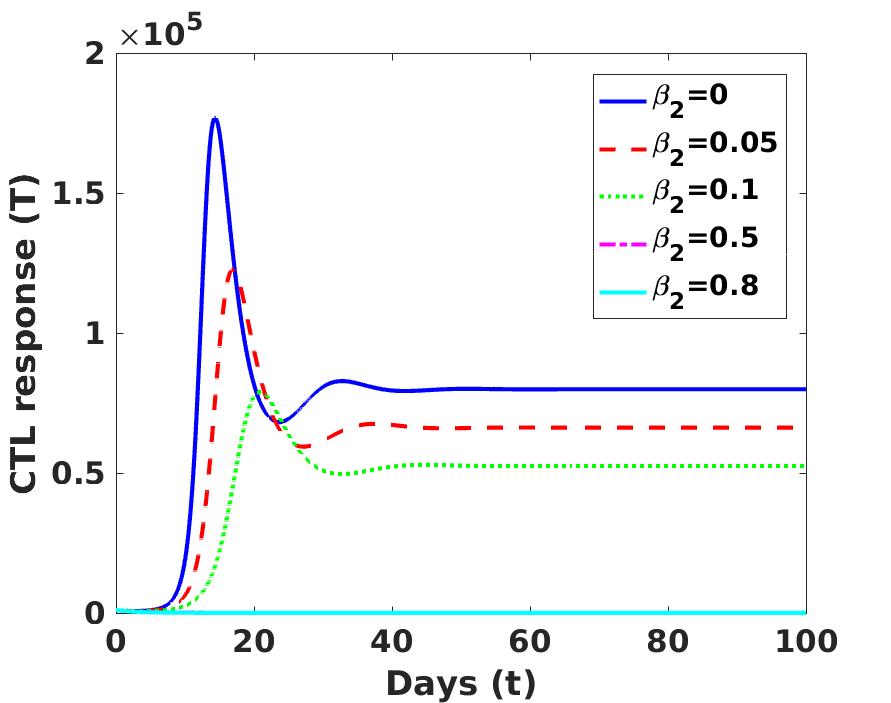}
\caption{}
\label{CML_Effect_of_Beta_2_C}
\end{subfigure}
\caption{Effect of $\beta_{2}$ on $C_{s}$, $C_{m}$ and $T$.}
\label{CML_Effect_of_Beta_2}
\end{figure}
\vspace{1cm}

\clearpage

\begin{figure}[!ht]
\centering
\begin{subfigure}[b]{0.25\textwidth}
\includegraphics[width=\textwidth,height=4cm]{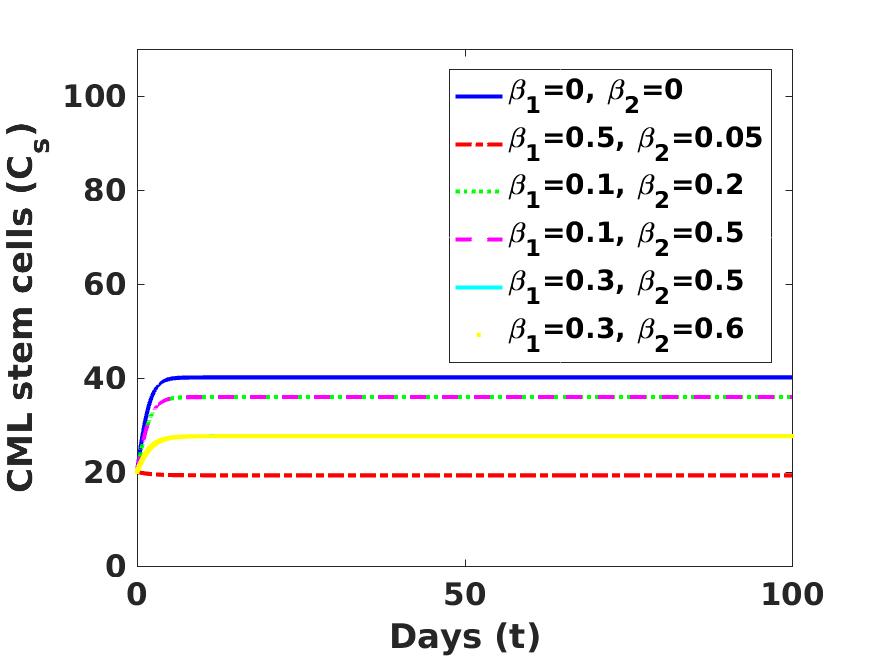}
\caption{}
\label{CML_Beta1_Beta2_Effect_A}
\end{subfigure}
\quad
\begin{subfigure}[b]{0.25\textwidth}
\includegraphics[width=\textwidth,height=4cm]{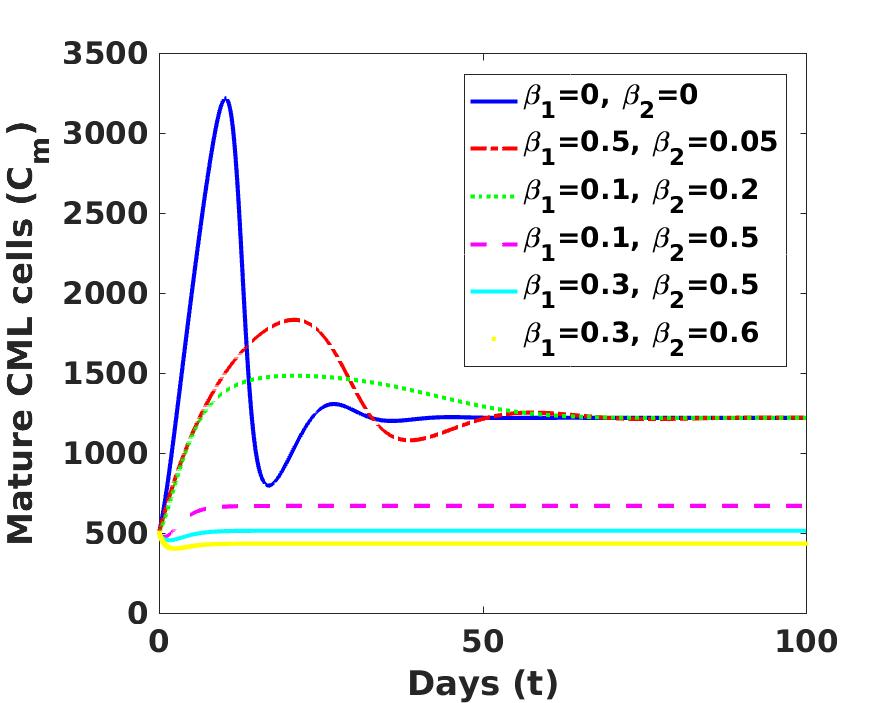}
\caption{}
\label{CML_Beta1_Beta2_Effect_B}
\end{subfigure}
\quad
\begin{subfigure}[b]{0.25\textwidth}
\includegraphics[width=\textwidth,height=4cm]{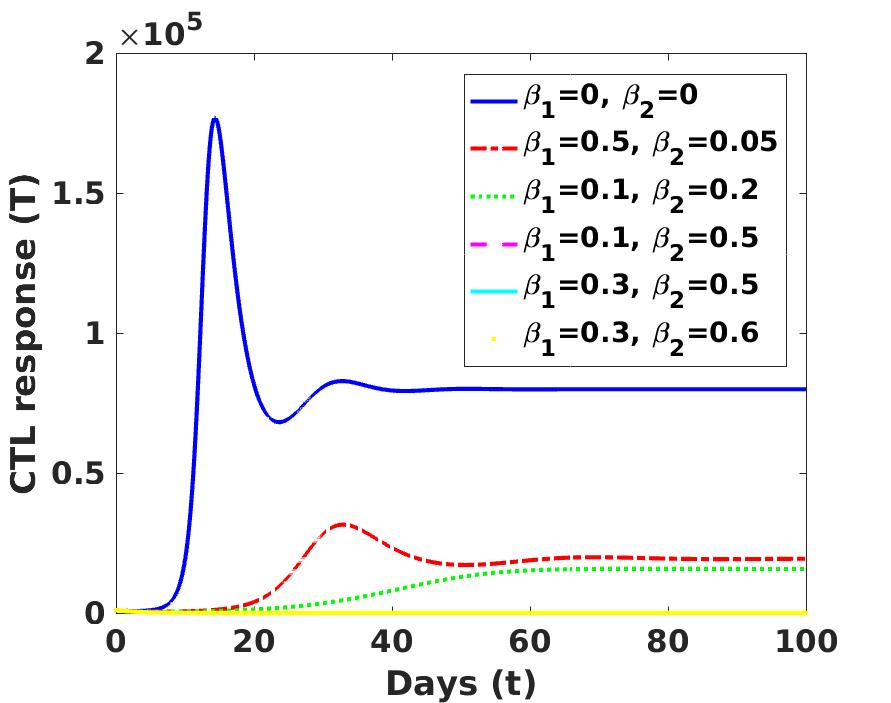}
\caption{}
\label{CML_Beta1_Beta2_Effect_C}
\end{subfigure}
\caption{Effect of dual therapy of $\beta_{1}$ and $\beta_{2}$ on $C_{s}$, $C_{m}$ and $T$.}
\label{CML_Beta1_Beta2_Effect}
\end{figure}
\vspace{1cm}

\begin{figure}[!ht]
\centering
\begin{subfigure}[b]{0.25\textwidth}
\includegraphics[width=\textwidth,height=4cm]{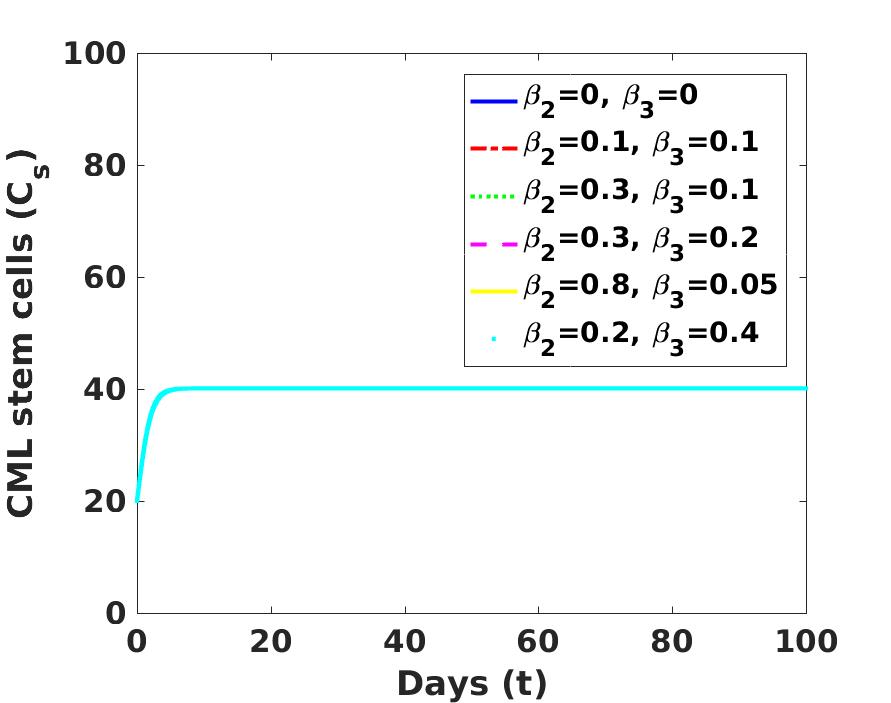}
\caption{}
\label{CML_Beta2_Beta3_Effect_A}
\end{subfigure}
\quad
\begin{subfigure}[b]{0.25\textwidth}
\includegraphics[width=\textwidth,height=4cm]{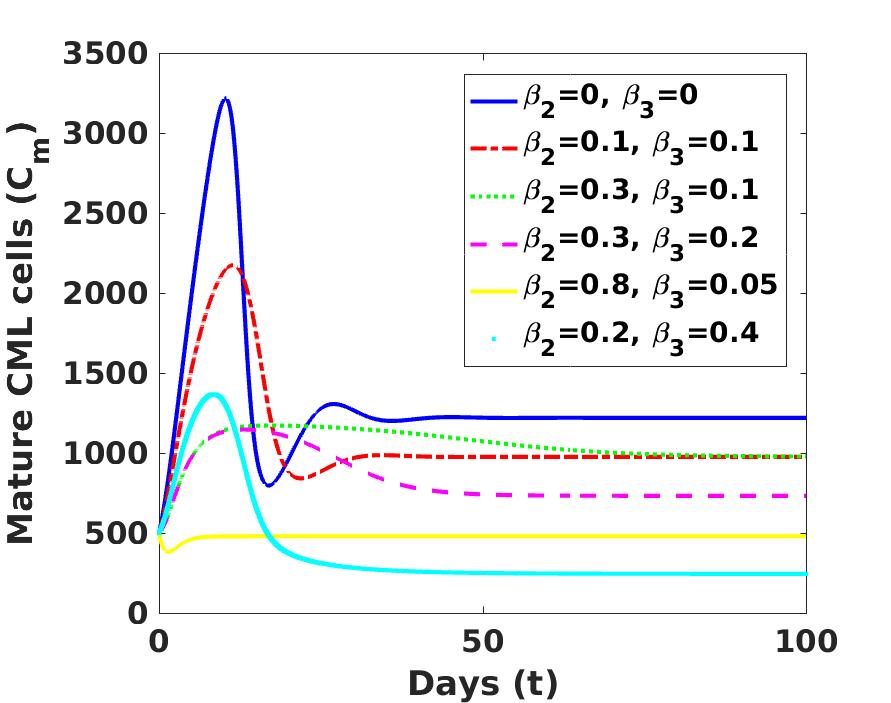}
\caption{}
\label{CML_Beta2_Beta3_Effect_B}
\end{subfigure}
\quad
\begin{subfigure}[b]{0.25\textwidth}
\includegraphics[width=\textwidth,height=4cm]{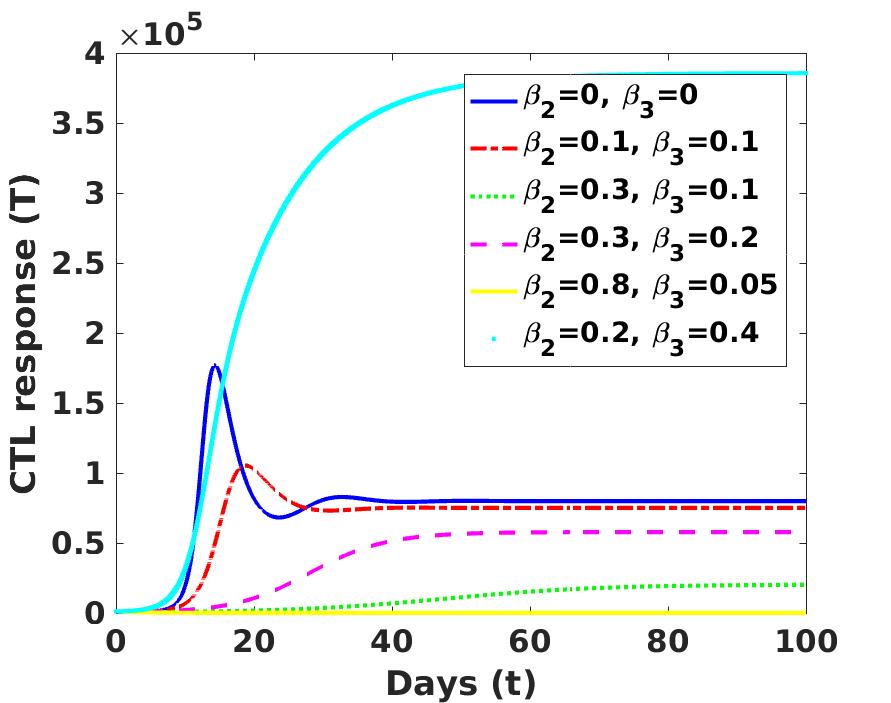}
\caption{}
\label{CML_Beta2_Beta3_Effect_C}
\end{subfigure}
\caption{Effect of dual therapy of $\beta_{2}$ and $\beta_{3}$ on $C_{s}$, $C_{m}$ and $T$.}
\label{CML_Beta2_Beta3_Effect}
\end{figure}
\vspace{1cm}

\begin{figure}[!ht]
\centering
\begin{subfigure}[b]{0.25\textwidth}
\includegraphics[width=\textwidth,height=4cm]{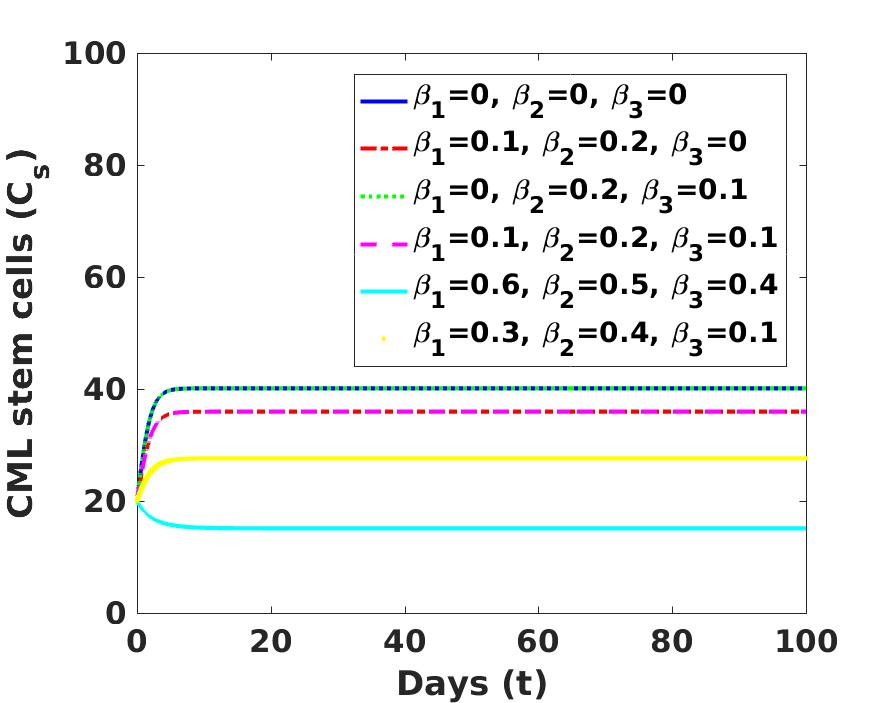}
\caption{}
\label{CML_Combination_Effect_A}
\end{subfigure}
\quad
\begin{subfigure}[b]{0.25\textwidth}
\includegraphics[width=\textwidth,height=4cm]{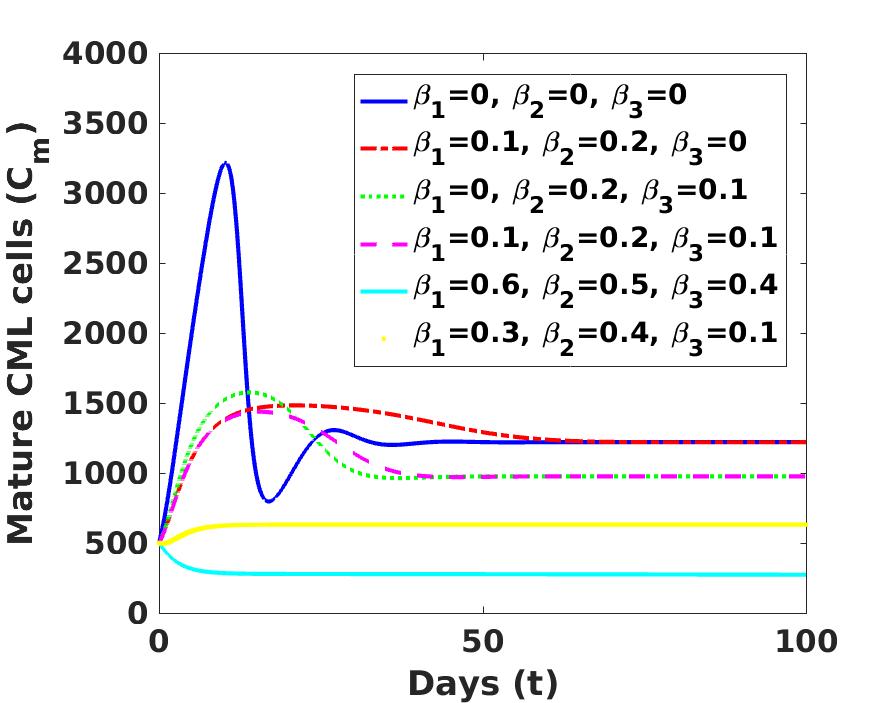}
\caption{}
\label{CML_Combination_Effect_B}
\end{subfigure}
\quad
\begin{subfigure}[b]{0.25\textwidth}
\includegraphics[width=\textwidth,height=4cm]{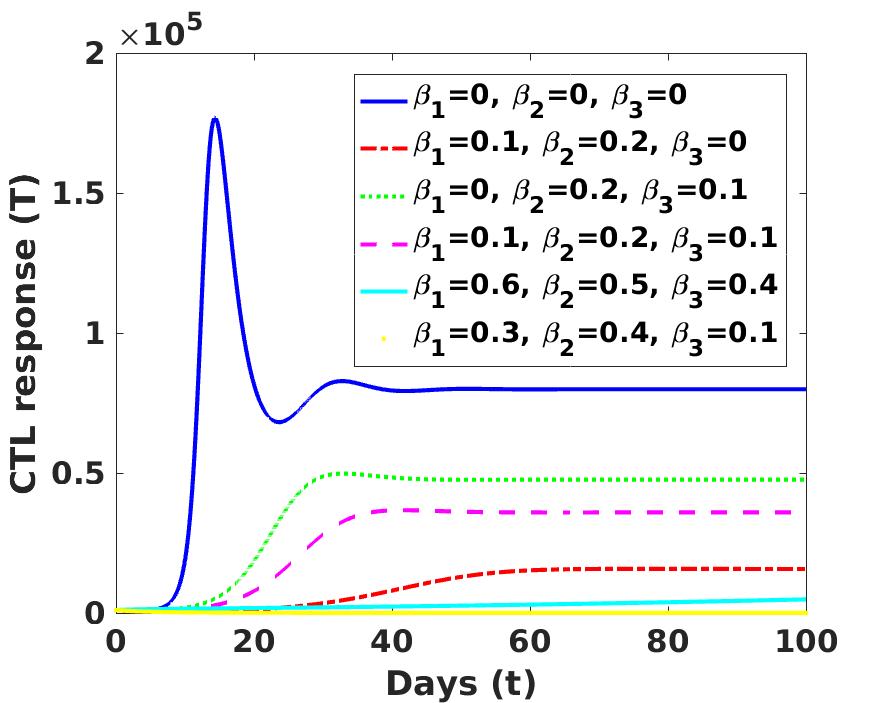}
\caption{}
\label{CML_Combination_Effect_C}
\end{subfigure}
\caption{Effect of combination therapy of $\beta_{1}$, $\beta_{2}$ and $\beta_{3}$ on $C_{s}$, $C_{m}$ and $T$.}
\label{CML_Combination_Effect}
\end{figure}
\vspace{1cm}
\clearpage

\begin{figure}[!ht]
\centering
\begin{subfigure}[b]{0.25\textwidth}
\includegraphics[width=\textwidth,height=4cm]{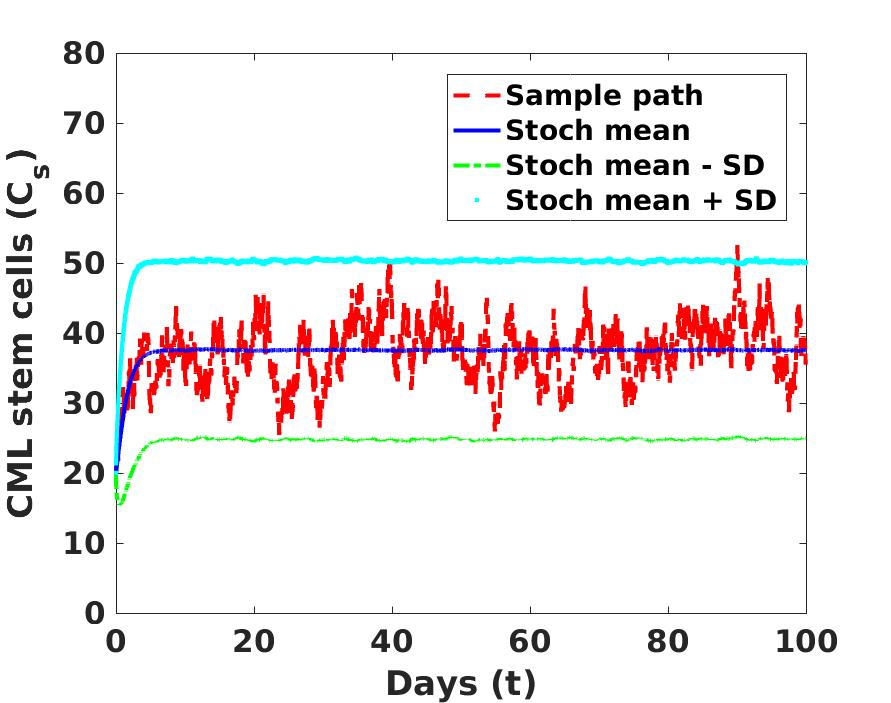}
\caption{}
\label{CML_SDE_A}
\end{subfigure}
\quad
\begin{subfigure}[b]{0.25\textwidth}
\includegraphics[width=\textwidth,height=4cm]{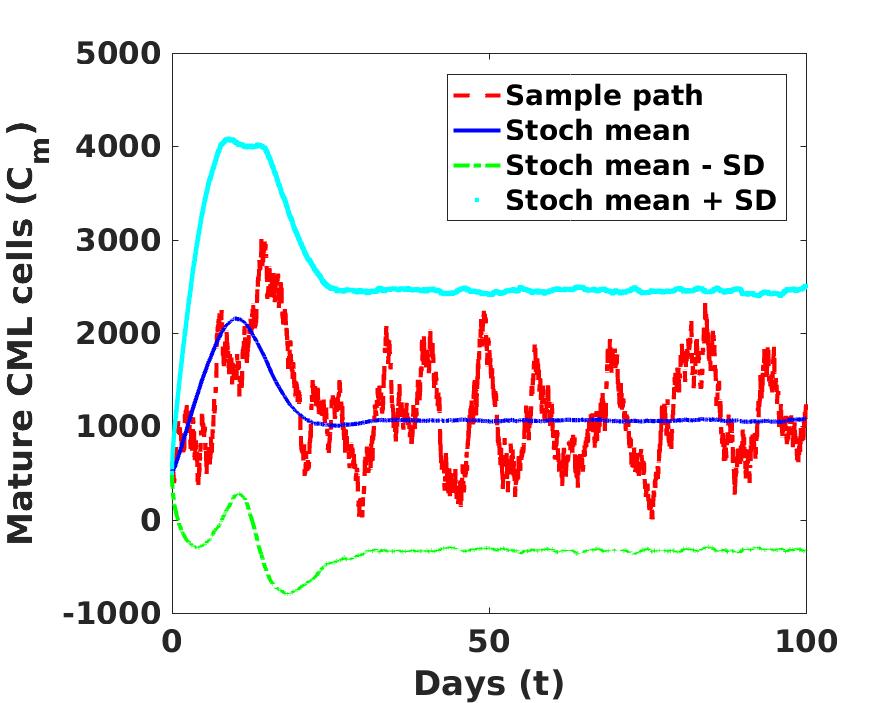}
\caption{}
\label{CML_SDE_B}
\end{subfigure}
\quad
\begin{subfigure}[b]{0.25\textwidth}
\includegraphics[width=\textwidth,height=4cm]{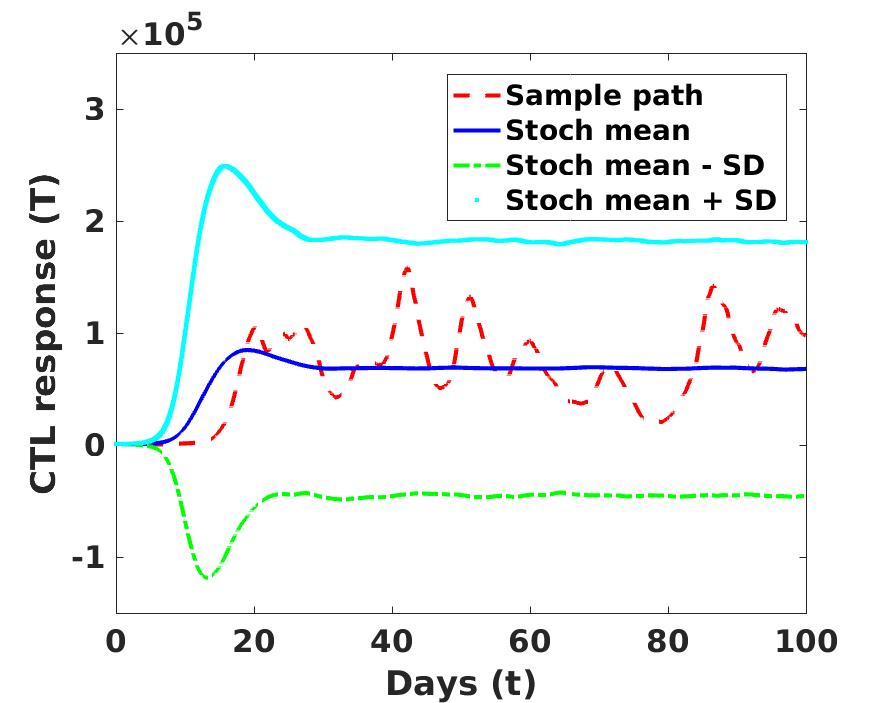}
\caption{}
\label{CML_SDE_C}
\end{subfigure}
\caption{Stochastic dynamics of $C_{s}$, $C_{m}$ and $T$.}
\label{CML_SDE_Figure}
\end{figure}
\vspace{1cm}

\begin{figure}[!ht]
\centering
\begin{subfigure}[b]{0.25\textwidth}
\includegraphics[width=\textwidth,height=4cm]{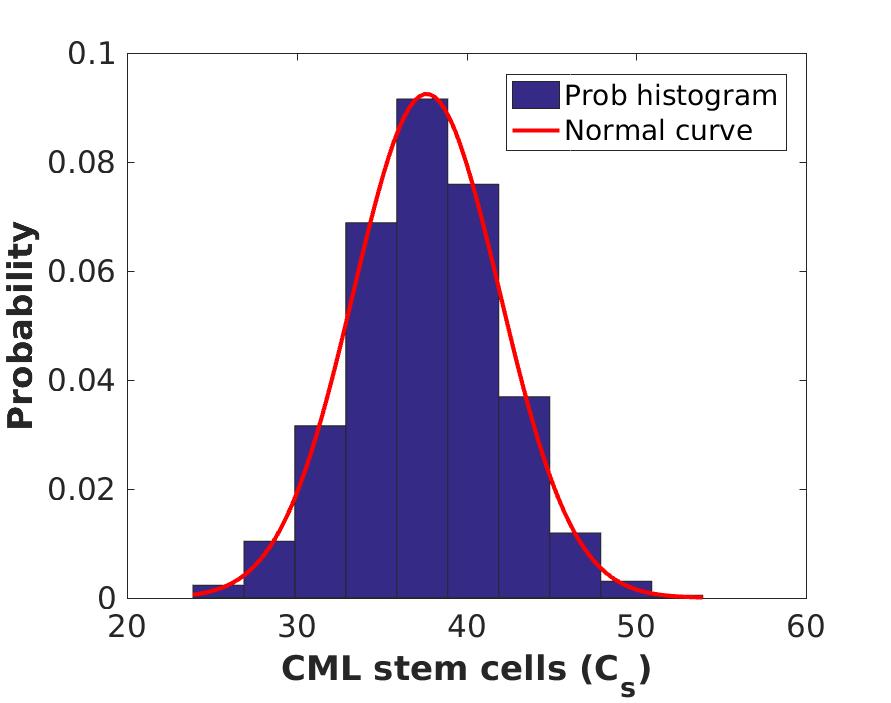}
\caption{}
\label{CML_Hist_A}
\end{subfigure}
\quad
\begin{subfigure}[b]{0.25\textwidth}
\includegraphics[width=\textwidth,height=4cm]{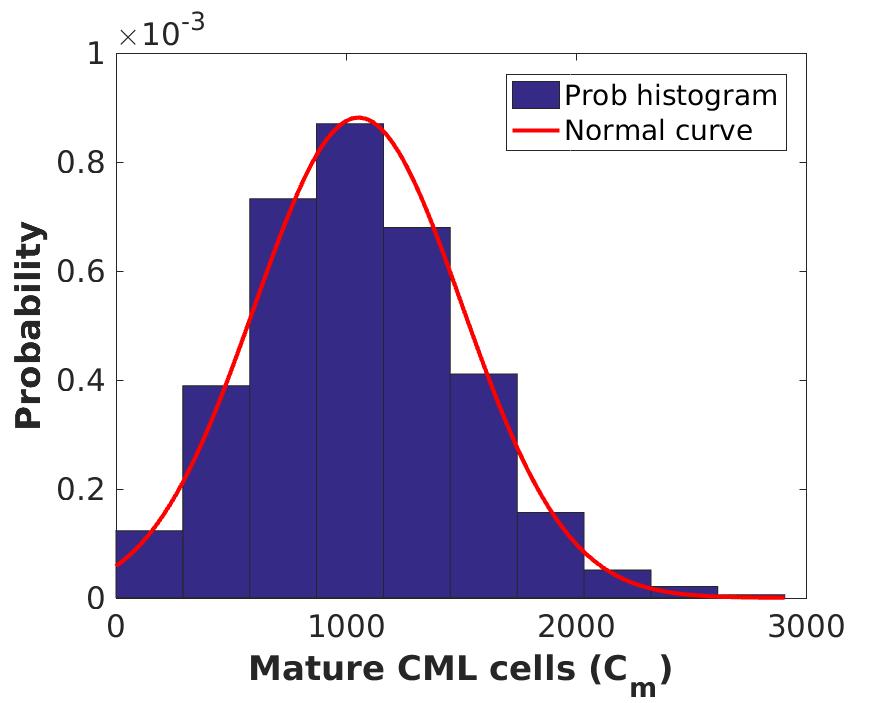}
\caption{}
\label{CML_Hist_B}
\end{subfigure}
\quad
\begin{subfigure}[b]{0.25\textwidth}
\includegraphics[width=\textwidth,height=4cm]{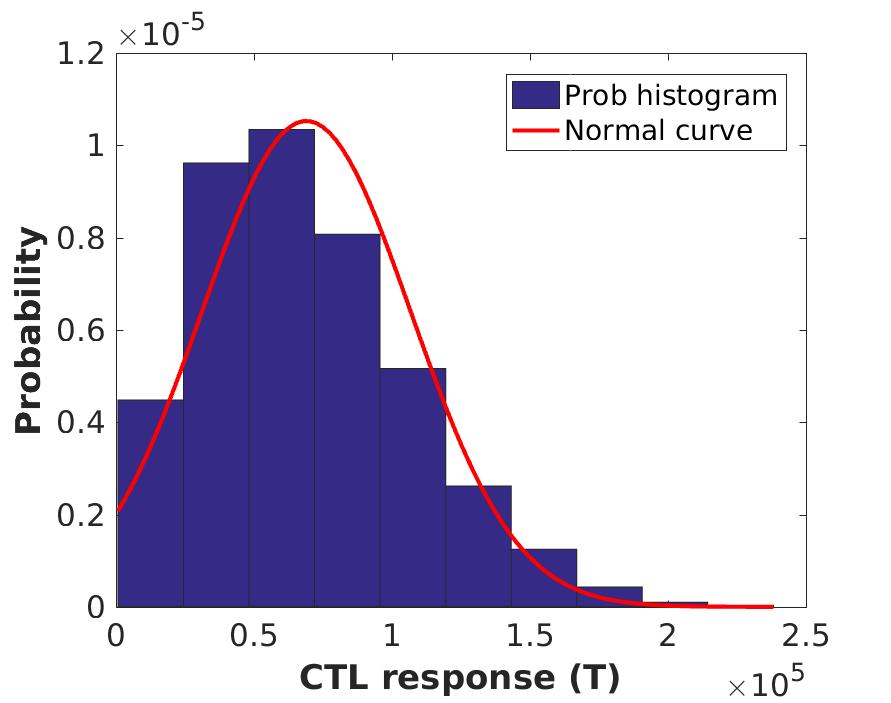}
\caption{}
\label{CML_Hist_C}
\end{subfigure}
\caption{Histograms of $C_{s}$, $C_{m}$ and $T$.}
\label{CML_Histogram}
\end{figure}
\vspace{1cm}

\begin{figure}[!ht]
\centering
\begin{subfigure}[b]{0.25\textwidth}
\includegraphics[width=\textwidth,height=4cm]{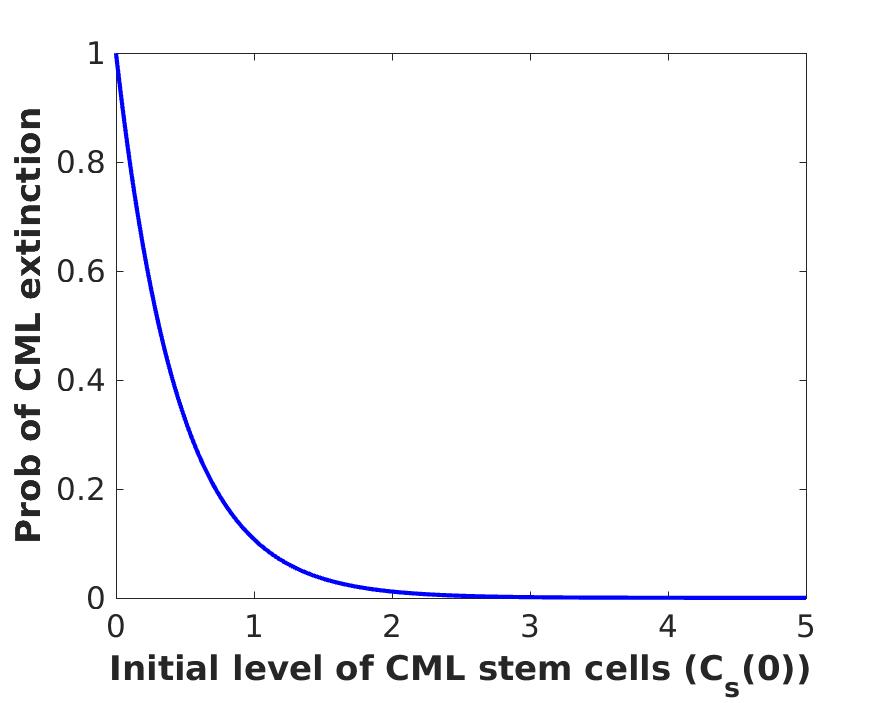}
\caption{}
\label{CML_Extinction_Cs}
\end{subfigure}
\quad
\begin{subfigure}[b]{0.25\textwidth}
\includegraphics[width=\textwidth,height=4cm]{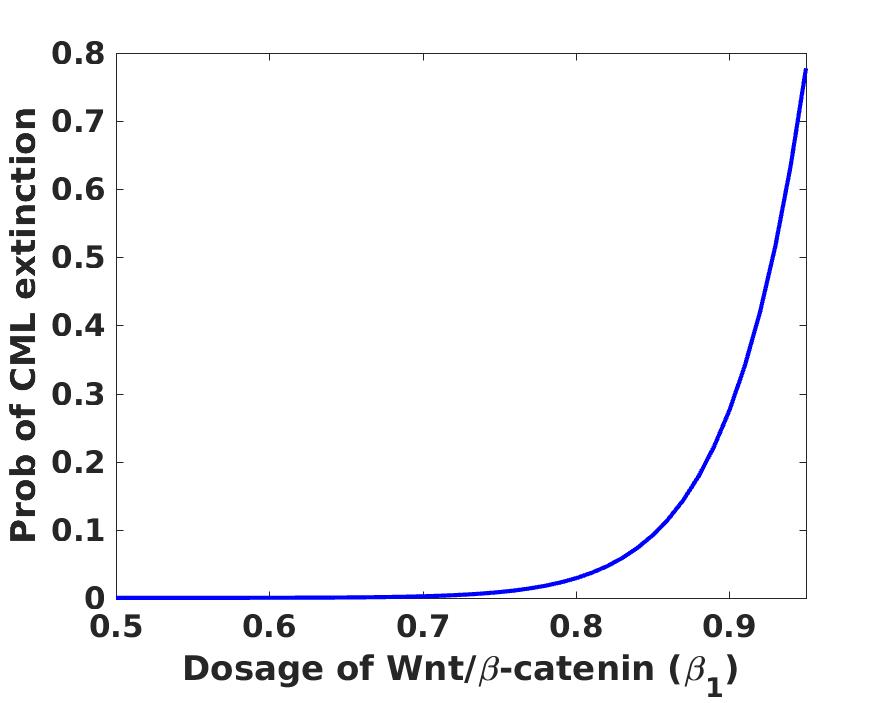}
\caption{}
\label{CML_Extinction_Beta1}
\end{subfigure}
\quad
\begin{subfigure}[b]{0.25\textwidth}
\includegraphics[width=\textwidth,height=4cm]{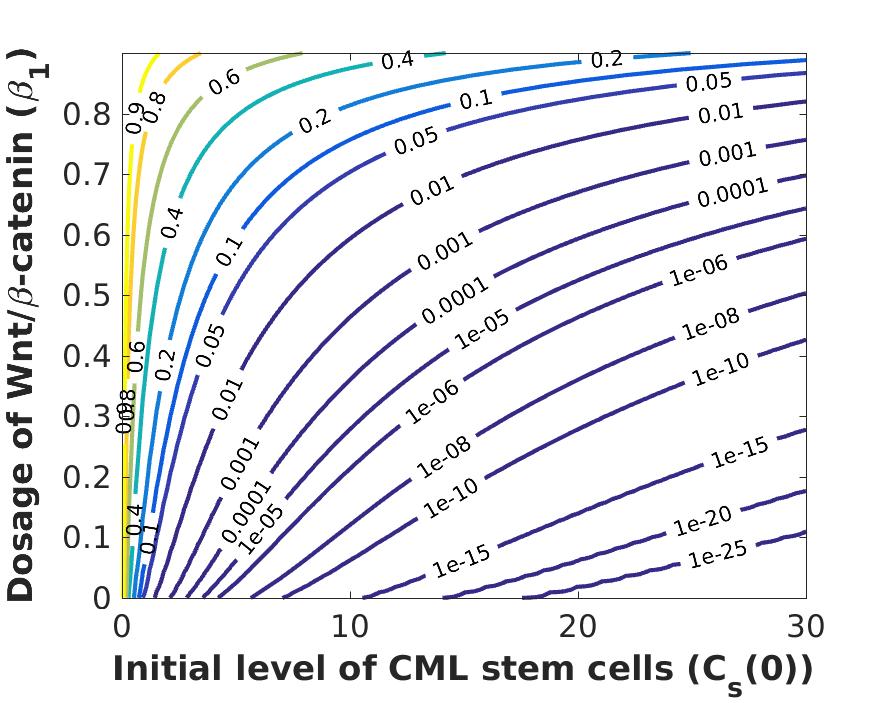}
\caption{}
\label{CML_Extinction_Cs_Beta1}
\end{subfigure}
\caption{The probability of CML extinction depending upon $C_s(0)$, $\beta_1$ and both.}
\label{CML_Extinction}
\end{figure}

\enlargethispage{1.4cm}

\end{document}